\documentclass[journal,onecolumn,draftclsnofoot]{IEEEtran}

\usepackage{amsmath,amssymb,amsfonts}
\usepackage{mathtools}
\usepackage{amsthm}
\usepackage{algorithmic}
\usepackage{graphicx}
\usepackage{pgfplots} 
\pgfplotsset{compat=1.15}

\usepackage{pgfgantt}
\usepackage{pdflscape}
\usepackage{pst-plot}
\usetikzlibrary{spy}
\usetikzlibrary{positioning,calc}
\usetikzlibrary{
calc,shapes,shapes.geometric,patterns}
\usetikzlibrary{shapes.multipart}
\usepackage{xfrac}
\usepackage{colortbl}
\usepackage{cancel}
\usetikzlibrary{arrows,positioning,calc,intersections}
\usetikzlibrary{datavisualization.formats.functions}
\def\BibTeX{{\rm B\kern-.05em{\sc i\kern-.025em b}\kern-.08em
    T\kern-.1667em\lower.7ex\hbox{E}\kern-.125emX}}
    \newcommand{\norm}[1]{\left\lVert#1\right\rVert}
\usepgfplotslibrary{fillbetween}
\usetikzlibrary{arrows}
\usepackage[caption=false,font=normalsize,labelfont=sf,textfont=sf]{subfig}
\usepackage{amssymb,bm}
\usepackage{float}
\usepackage{graphicx}
\usepackage{amsmath}
\usepackage{hyperref}
\usepackage{multirow}
\usepackage{xcolor}

\usepackage{tikz}
\usepackage{mathrsfs}
\usepackage{tikz}
\usepackage[american,siunitx]{circuitikz}
\usetikzlibrary{arrows,calc,positioning}
\usepackage{bbm}

\usepackage{cite} 

\usepackage{comment}

\renewcommand{\epsilon}{\varepsilon}

\newcommand{\C}{\mathcal{C}}

\tikzset{ar/.style={-latex,shorten >=-1pt, shorten <=-1pt}}

\newcommand\underrel[3][]{\mathrel{\mathop{#3}\limits_{%
      \ifx c#1\relax\mathclap{#2}\else#2\fi}}}

\def \A{\mathcal{A}}
\def \B{\mathcal{B}}
\def \C{\mathcal{C}}
\def \D{\mathcal{D}}
\def \E{\mathcal{E}}
\def \F{\mathcal{F}}
\def \G{\mathcal{G}}
\def \H{\mathcal{H}}

\def \N{\mathcal{N}}

\def \S{\mathcal{S}}

\def \U{\mathcal{U}}

\def \Y{\mathcal{Y}}

\def \sG {\mathscr{G}}

\def \fG{\mathbf{G}}

\def \fa{\mathbf{a}}
\def \fb{\mathbf{b}}
\def \fg{\mathbf{g}}

\def \fu{\mathbf{u}}

\def \fx{\mathbf{x}}
\def \fy{\mathbf{y}}
\def \fz{\mathbf{z}}

\def \fX{\mathbf{X}}
\def \fY{\mathbf{Y}}
\def \fZ{\mathbf{Z}}
\def \f0{\mathbf{0}}

\definecolor{blau_1a}{RGB}{93,133,195}
\definecolor{blau_2a}{RGB}{0,156,218}
\definecolor{gruen_3a}{RGB}{80,182,149}
\definecolor{gruen_4a}{RGB}{175,204,80}
\definecolor{gruen_5a}{RGB}{221,223,72}
\definecolor{orange_6a}{RGB}{255,224,92}
\definecolor{orange_7a}{RGB}{248,186,60}
\definecolor{rot_8a}{RGB}{238,122,52}
\definecolor{rot_9a}{RGB}{233,80,62}
\definecolor{lila_10a}{RGB}{201,48,142}
\definecolor{lila_11a}{RGB}{128,69,151}

\definecolor{blau_1b}{RGB}{0,90,169}
\definecolor{blau_2b}{RGB}{0,131,204}
\definecolor{gruen_3b}{RGB}{0,157,129}
\definecolor{gruen_4b}{RGB}{153,192,0}
\definecolor{gruen_5b}{RGB}{201,212,0}
\definecolor{orange_6b}{RGB}{253,202,0}
\definecolor{orange_7b}{RGB}{245,163,0}
\definecolor{rot_8b}{RGB}{236,101,0}
\definecolor{rot_9b}{RGB}{230,0,26}
\definecolor{lila_10b}{RGB}{166,0,132}
\definecolor{lila_11b}{RGB}{114,16,133}

\definecolor{mycolor1}{RGB}{249,245,233}

\usepackage{accents}

\theoremstyle{remark}	\newtheorem{theorem}{Theorem}
\theoremstyle{remark}	\newtheorem{lemma}[theorem]{Lemma}
\theoremstyle{remark}	
\theoremstyle{remark}	
\theoremstyle{remark} \newtheorem{definition}{Definition}
\theoremstyle{remark} \newtheorem{remark}{Remark}
\theoremstyle{remark}

\usepackage{tikz}
\usetikzlibrary{arrows}

\usepackage{circuitikz}
\usetikzlibrary{positioning}

\usepackage{circuitikz}
\usetikzlibrary{arrows.meta,positioning}

\usepackage{comment}

\usetikzlibrary{fit,backgrounds}

\usepackage{autobreak}
\allowdisplaybreaks
\begin{document}
\title{Deterministic Identification Over Fading Channels}

\author{
	\vspace{0.1cm}
    \IEEEauthorblockN{Mohammad J. Salariseddigh\IEEEauthorrefmark{1}, Uzi Pereg\IEEEauthorrefmark{1}, 
    Holger Boche\IEEEauthorrefmark{2}, and Christian Deppe\IEEEauthorrefmark{1}} \\
		\vspace{0.25cm}
    \IEEEauthorblockA{\normalsize \IEEEauthorrefmark{1}Institute for Communications Engineering, Technical University of Munich \\
    \IEEEauthorrefmark{2}Chair of Theoretical Information Technology, Technical University of Munich\\
    {\tt Email}: {\{mohammad.j.salariseddigh, uzi.pereg, boche, christian.deppe\}@tum.de}}




%
}
\maketitle
\begin{abstract}
Deterministic identification (DI) is addressed for Gaussian channels with fast and slow fading, where channel side information is available at the decoder. In particular, it is established that the number of messages scales as $2^{n\log(n)R}$, where $n$ is the block length and $R$ is the coding rate. Lower and upper bounds on the DI capacity are developed in this scale for fast and slow fading. Consequently, the DI capacity is infinite in the exponential scale and zero in the double-exponential scale, regardless of the channel noise.
\end{abstract}
\begin{IEEEkeywords}
Channel capacity, fading channels, identification, identification without randomization, deterministic codes, super exponential growth,  channel side information.
\end{IEEEkeywords}
\IEEEpeerreviewmaketitle
\section{Introduction}
Modern communications require the transfer of enormous amounts of data in wireless systems, for cellular communication \cite{BCCK09}, sensor networks \cite{S08}, smart appliances \cite{TLTWYZLNTH19}, and the internet of things \cite{DHV17}, etc. 
Wireless communication is often modelled by fading channels with additive white Gaussian noise \cite{LYHW18,G05,OSW94,BPS98,ZCL09,M18,YX19,ND20,AG20}. In the fast fading regime, the transmission spans over a large number of coherence time intervals \cite{TV05}, hence the signal attenuation is characterized by a stochastic process or a sequence of random parameters \cite{LYT10,HK19,KJ20,NL20}. In some applications, the receiver may acquire channel side information (CSI)  by instantaneous estimation of the channel parameters \cite{GV97,LT07,VMAA14}. On the other hand, in the slow fading regime, the latency is short compared to the coherence time \cite{TV05}, and the behaviour is that of a compound channel \cite{CK82,S97,SS03,BD06,KMJY19}.

In the fundamental communication paradigm considered by
Shannon \cite{S48},  a sender wishes to convey a message through a noisy channel in such a manner that  the receiver will be able to retrieve the original message. In other words, the decoder's task is to determine which message was sent.  Ahlswede and Dueck \cite{AD89} introduced a scenario of a different nature, where the decoder only performs identification and determines whether a particular message was sent or not \cite{AD89,HanBook,SCR20-2}.
Applications include 
the tactile internet \cite{TactileInternet14}, vehicle-to-X communications \cite{KBMRAZT16,JVNRCR16},
digital watermarking \cite{M01,S01,AN06}, online sales \cite{GCE07,GC08}, industry 4.0 \cite{
Industry4.0-3}, health care \cite{Healthcare04-Patent}, and other event-triggered systems.

%
We give two motivating examples for applications of the identification paradigm.
Molecular communication is a promising contender for future applications such as the sixth generation of cellular communication (6G) \cite{6G+,6G_PST}, in which a number of applications require alerts to be identified \cite{6G_PST}. Furthermore, in other systems of molecular communication, a nano-device needs to determine the occurrence of a specific event. For example, in the course of targeted drug delivery \cite{Muller04,Nakano13} or cancer treatment \cite{Hobbs_ea98,Jain99,Wilhelm16}, a nano-device will seek to know whether the blood pH exceeds a critical threshold or not, whether a specific drug is released or not, whether another nano-device has replicated itself, whether a certain molecule was detected, whether a target location in the vessels is identified, or whether the molecular storage is empty, etc \cite{Nakano14}. A second application for identification is  vehicle-to-X communications, where a vehicle that collects  sensor data may ask whether a certain alert message concerning the future movement of an adjacent vehicle was transmitted or not \cite[Sec.~VII]{BD17}.

The identification problem \cite{AD89} can be regarded as a \emph{Post-Shannon} \cite{CGC03} model where the decoder does not perform an estimation, but rather a binary hypothesis test to decide between the hypotheses `sent' or `not sent', based on the observation of the channel output. As the sender has no knowledge of the desired message that the receiver is interested in, the identification problem can be regarded as a test of many hypotheses occurring simultaneously. The scenario where the receiver misses and does not identify his message is called a type I error, or `missed identification', whereas the event where the receiver accepts a false message is called a type II error, or `false identification'.

Ahlswede and Dueck \cite{AD89} required randomized coding for their identification-coding scheme. This means that a randomized source is available to the sender. The sender can make his encoding dependent on the output of this source. It is known that this resource cannot be used to increase the transmission capacity of discrete memoryless channels \cite{A78}.
A remarkable result of identification theory is that given local randomness at the encoder, reliable identification can be attained such that the code size, i.e., the number of messages, grows double exponentially in the block length $n$, i.e., $\sim 2^{ 2^{nR}}$ \cite{AD89}. This differs sharply from the traditional transmission setting where the code size scales only exponentially, i.e., $\sim{2^{nR}}$. Beyond the exponential gain in identification, the extension of the problem to more complex scenarios reveals that the identification capacity has a very different behavior compared to the transmission capacity \cite{feedback,correlation,BV00,BD18_2,W04,BL17}. For instance, feedback \emph{can} increase the identification capacity \cite{feedback} of a memoryless channel, as opposed to the transmission capacity \cite{S56}. Nevertheless, it is difficult to implement randomized-encoder identification (RI) codes that will achieve such performance, because it requires the encoder to process a bit string of exponential length. The construction of identification codes is considered in \cite{SCR20-2,VK93,KT99,Bringer09,Bringer10}. Identification for Gaussian channels is considered in \cite{MasterThesis,LDB20,Labidi2021,Ezzine2021,BV00}.

In the deterministic setup for a DMC, the number of messages scales exponentially in the blocklength \cite{AD89,AN99,J85,Bur00}, as in the traditional setting of transmission. Nevertheless, the achievable identification rates are significantly higher than those of transmission. In addition, deterministic codes often have the advantage of simpler implementation and simulation \cite{PP09}, explicit construction \cite{A09}, and single-block reliable performance. In particular, J\'aJ\'a \cite{J85} showed that the deterministic identification (DI) capacity \begin{footnote}{The DI capacity in the literature is also referred to as the non-randomized identification (NRI) capacity \cite{AN99} or the dID capacity \cite{BV00}.}\end{footnote} of a binary symmetric channel is 1 bit per channel use, based on combinatorial methods. The meaning of this result is that one can exhaust the entire input space and assign (almost) all sequences in the $n$-dimensional space $\{0,1\}^n$ as codewords. Ahlswede et al. \cite{AD89,AN99} stated that the DI capacity of a DMC  with a stochastic matrix $W$ in the exponential scale is given by the logarithm of the number of distinct row vectors of $W$ (see \cite[Sec.~IV]{AD89} and abstract of \cite{AN99}). Furthermore, the DI $\epsilon$-capacity of the Gaussian channel was addressed by Burnashev \cite{BV00}.

In a recent work by the authors \cite{SPBD21ICC,SPBD20arXiv}, we addressed deterministic identification for the DMC subject to an input constraint and have also shown that the DI capacity of the standard Gaussian channel, without fading, is infinite in the exponential scale. Our previous results \cite{SPBD21ICC,SPBD20arXiv} reveal a gap of knowledge in the following sense. For a finite blocklength $n$, the number of codewords must be finite.  Thereby, the meaning of the infinite capacity result is that the number of messages scales super-exponentially. The question remains what is the true order of the code size.
In mathematical terms, what is the scale $L$ for which the DI capacity is positive yet finite. Here, we will answer this question.

In this paper, we consider deterministic identification for Gaussian channels with fast fading and slow fading, where channel side information (CSI) is available at the decoder. We show that for Gaussian channels, the number of messages scales as $2^{n\log(n)R}$, and develop lower and upper bounds on the DI capacity in this scale. As a consequence, we deduce  that the DI capacity of a Gaussian Channel with fast fading is infinite in the exponential scale, and zero in the double-exponential scale, regardless of the channel noise. For slow fading, the DI capacity in the exponential scale is infinite, unless  the fading gain can be zero or arbitrarily close to zero (with positive probability), in which case the DI capacity is zero. In comparison with the double exponential scale in RI coding, the scale here is significantly lower.

The results have the following geometric interpretation. At first glance, it may seem reasonable that for the purpose of identification, one codeword could represent two messages. While identification allows overlap between decoding regions \cite{AADT20,LDB20}, overlap at the encoder is not allowed for deterministic codes. We observe that when two messages are represented by codewords that are close to one another, then identification fails. Thus, deterministic coding imposes the restriction that the codewords need to be distanced from each other.

Based on fundamental properties of packing arrangements \cite{C10}, \cite{CHSN13}, 
the optimal packing of non-overlapping spheres of radius $\sqrt{n\epsilon}$ contains an exponential number of spheres, and by decreasing the radius of the codeword spheres, the exponential rate can be made arbitrarily large. However, in the  derivation of our lower bound in the $2^{n\log(n)R}$-scale, we pack spheres of a sub-linear radius $\sqrt{n\epsilon_n}\sim n^{1/4}$, which results in $\sim 2^{\frac{1}{4}n\log(n)}$ codewords.
\section{Definitions and Related Work}
\label{sec:preliminaries}
In this section, we introduce the channel models and coding definitions. Here, we only consider the Gaussian channel with fast fading. The channel description and coding definition for slow fading will be presented in Section~\ref{Sec.GaussianChannelSlow}.
\subsection{Notation}
 We use the following notation conventions throughout. Lowercase letters $x,y,z,\ldots$  stand for constants and values of random variables, and uppercase letters $X,Y,Z,\ldots$ stand for random variables.  
 The distribution of a real random variable $X$ is specified by a cumulative distribution function (cdf)
$F_X(x)=\Pr(X\leq x)$ for $x\in\mathbb{R}$, or alternatively, by a probability density function (pdf) $f_X(x)$, when it exists.
 The notation $\mathbf{x}=(x_1,x_{2},\ldots,x_n)$ is used  for a sequence of length $n$, and  the $\ell^2$-norm of $\mathbf{x}$ is denoted  by $\norm{\mathbf{x}}$. Element-wise product of vectors is denoted by $\fx\circ\fy=(x_t y_t)_{t=1}^n$. A random sequence $\fX$ and its distribution $F_{\fX}(\fx)$ are defined accordingly. 
We denote the hyper-sphere of radius 
$r$ around $\fx_0$ by
\begin{align}
    \S_{\fx_0}(n,r) = \left\{\fx\in\mathbb{R}^n \;:\, \norm{\fx-\fx_0} \leq r \right\} \;,\,
\end{align}
and its volume by $\text{Vol}(\S)$. The closure of a set $\A$ is denoted by $\text{cl}(\A)$.
The set of consecutive natural numbers from $1$ to $M$ is denoted by $[\![M]\!]$.
\subsection{Fast Fading Channel}
\label{Subsec.ChannelDescription}
Consider the Gaussian channel $\mathscr{G}_{\,\text{fast}}$ with fast fading, specified by the input-output relation
\begin{align}
    \mathbf{Y}=\mathbf{G}\circ\mathbf{x} + \fZ \;,\,
\end{align}
where $\mathbf{G}$ is a random sequence of fading coefficients and $\fZ$ is an additive white Gaussian process (see Figure~\ref{Fig.GaussianChannelFast}).
Specifically, $\mathbf{G}$ is a 
sequence of i.i.d. continuous random variables $ \sim f_G$ with finite moments, while the noise sequence $\fZ$ is i.i.d. $\sim \mathcal{N}(0,\sigma_Z^2)$. It is assumed that the noise sequence $\fZ$ and the sequence of fading coefficients $\fG$ are statistically independent, and that the values of the fading coefficients belong to a bounded set $\G$, either countable or uncountable. The transmission power is limited to 
$
    \norm{\textbf{x}}^2\leq nA 
$. 

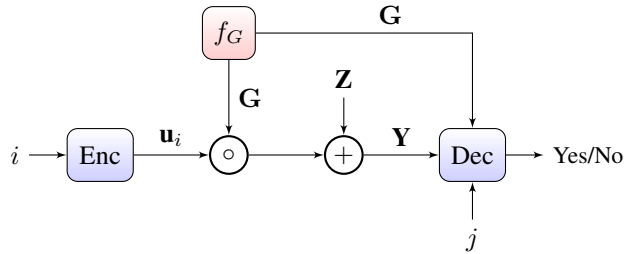
\begin{figure}[htb]
    \centering
	\tikzstyle{l} = [draw, -latex']
\tikzstyle{farbverlauf} = [ top color=white, bottom color=white!80!gray]
\tikzstyle{block1} = [draw,top color=white, bottom color=blue!20!white, rectangle, rounded corners,
minimum height=2em, minimum width=2.5em]
\tikzstyle{block2} = [draw,top color=white, bottom color=white!80!blue, rectangle, rounded corners,
minimum height=2em, minimum width=2.5em]
\tikzstyle{block3} = [draw,top color=white, bottom color=red!20, rectangle, rounded corners,
minimum height=2em, minimum width=2em]
\tikzstyle{input} = [coordinate]
\tikzstyle{sum} = [draw, circle,inner sep=0pt, minimum size=5mm,  thick]
\tikzstyle{multiplexer} = [draw, circle,inner sep=0pt, minimum size=5mm,  thick]
\tikzstyle{arrow}=[draw,->]
\begin{tikzpicture}[auto, node distance=2cm,>=latex']
\node[] (M) {$i$};
\node[block1,right=.5cm of M] (enc) {Enc};
\node[multiplexer, right=1cm of enc] (multiplexer) {$\circ$};
\node[sum, right=1cm of multiplexer] (channel) {$+$};

\node[block2, right=1cm of channel] (dec) {Dec};
\node[below=.5cm of dec] (Target) {$j$};
\node[right=.5cm of dec] (Output) {$\text{\small Yes/No}$};

\node[block3,above=1cm of multiplexer] (gain) {$f_G$};

\path [l] (gain) -| node [text width=2.5cm,above] {$\fG$} (dec) ;

\node[above=.5cm of channel] (noise) {$\textbf{Z}$};

\draw[->] (M) -- (enc);
\draw[->] (enc) --node[above]{$\textbf{u}_i$} (multiplexer);
\draw[->] (multiplexer) -- (channel);
\draw[->] (gain) -- node [] {$\fG$} (multiplexer);
\draw[->] (noise) -- (channel);
\draw[->] (channel) --node[above]{$\textbf{Y}$} (dec);

\draw[->] (dec) -- (Output);
\draw[->] (Target) -- (dec);

\end{tikzpicture}
	\caption{Deterministic identification for the Gaussian channel with fast fading, where $\mathbf{G}$ is a 
sequence of i.i.d. fading coefficients $ \sim f_G$, and the noise sequence $\fZ$ is i.i.d. $\sim \mathcal{N}(0,\sigma_Z^2)$. }
	\label{Fig.GaussianChannelFast}
\end{figure}
\subsection{Coding with Fast Fading}
In this paper, we consider codes with different size orders. For instance, when we discuss the exponential scale, we refer to a code size that scales as
$L(n,R)=2^{nR}$. On the other hand, in the double-exponential scale, the code size is 
$L(n,R)=2^{2^{nR}}$.
\begin{definition}
    \label{Def.DominatingScale}
    Let $L_1(n,R)$ and $L_2(n,R)$ be two coding scales. We say that $L_1$ \emph{dominates} $L_2$ if
    \begin{align}
        \lim_{n\rightarrow\infty}\frac{L_2(n,b)}{L_1(n,a)} = 0 \;.\,
    \end{align}
   for all $a,b>0$. We will denote this relation by 
    $L_2\prec L_1$.
\end{definition}
In complexity theory of computer science, the relation above is denoted by the `small $o$-notation',
$L_2(n,1)=o(L_1(n,1))$ \cite{CLRS09}. Beyond exponential, other orders that commonly appear in complexity theory are the linear, logarithmic, and polynomial scales, $nR$, $\log(nR)$, and $(nR)^k$. Here, we show that the scale of the DI capacity turns out to be the $L(n,R)=n^{nR}=2^{n\log(n)R}$.
The corresponding ordering is
\begin{align}
    nR \prec \log(nR) \prec (nR)^k \prec 2^{nR} \prec 2^{n\log(n)R}  \prec 2^{2^{nR}}.
\end{align}
We note that the scale of the DI capacity can be viewed is a special case of a tetration function, as ${^{2}n}=n^n=2^{n\log(n)}$ \cite{G47,AB09}.

\begin{definition}
\label{Def.GdeterministicIDCode}
An $(L(n,R),n)$ DI code with channel side information (CSI) at the decoder for a Gaussian channel $\sG_{\,\text{fast}}$ under input constraint $A$, assuming $L(n,R)$ is an integer, is defined as a system $(\U,\mathscr{D})$ which consists of a codebook $\U=\{ \mathbf{u}_i\}$ for $i\in[\![L(n,R)]\!]$ where $\U\subset \mathbb{R}^n$, such that 
\begin{align} 
    \norm{\mathbf{u}_i}^2\leq nA \;,\,
\end{align}
for all $i\in[\![L(n,R)]\!]$ and a collection of decoding regions
\begin{align}
    \mathscr{D}=\{ \D_{i,\mathbf{g}} \} \;,\,
\end{align}
for $i\in[\![L(n,R)]\!]$ and $\mathbf{g}\in \G^n$ with
\begin{align}
    \bigcup_{i=1}^{L(n,R)} \D_{i,\mathbf{g}} \subset \mathbb{R}^n \;.\,
\end{align}
Given a message $i\in [\![L(n,R)]\!]$, the encoder transmits $\mathbf{u}_i$. The decoder's aim is to answer the following question: Was a desired message $j$ sent or not?
There are two types of errors that may occur: Rejecting the true message, or accepting a false message. Those are referred to as type I and type II errors, respectively.
 
 The error probabilities of the identification code $(\U,\mathscr{D})$ are given by
\begin{align}
 P_{e,1}(i)&= 1-\int_{\G^n}  f_{\mathbf{G}}(\mathbf{g})\bigg[ \int_{\D_{i,\mathbf{g}}} f_{\fZ}(\mathbf{y}-\mathbf{g}\circ\mathbf{u}_i) d\mathbf{y}\bigg] d\mathbf{g} \;,\, 
 \label{Eq.GTypeIErrorDefFast}
 \\
 P_{e,2}(i,j)&=\int_{\G^n} f_{\mathbf{G}}(\mathbf{g}) \bigg[
 \int_{\D_{j,\mathbf{g}}} 
 f_{\fZ}(\mathbf{y}-\mathbf{g}\circ\mathbf{u}_i) \, d\mathbf{y} \bigg] d\mathbf{g} \;,\,
 \label{Eq.GTypeIIErrorDefFast}
\end{align}
with $f_{\fZ}(\mathbf{z})=\frac{1}{(2\pi\sigma_Z^2)^{n/2}} e^{-\norm{\mathbf{z}}^2/2\sigma_Z^2}$ (see Figure~\ref{Fig.GaussianChannelFast}).
An $(L(n,R),n,\lambda_1,\lambda_2)$ DI code further satisfies
\begin{align}
\label{Eq.GTypeIError}
 P_{e,1}(i) &\leq \lambda_1 \;,\, \\
\label{Eq.GTypeIIError}
 P_{e,2}(i,j) &\leq \lambda_2  \;,\,
\end{align}
for all $i,j\in [\![L(n,R)]\!]$, such that
$i\neq j$.
A rate $R>0$ is called achievable if for every
$\lambda_1,\lambda_2>0$ and sufficiently large $n$, there exists an ($L(n,R),n,\lambda_1,\lambda_2$) DI code.
The operational DI capacity in the $L$-scale is defined as the supremum of achievable rates, and will be denoted by $\mathbb{C}_{DI}(\mathscr{G}_{\,\text{fast}},L)$.
\end{definition}
In \cite{AD89}, Ahlswede and Dueck presented identification codes with a \emph{double}-exponential number of messages, i.e., $L(n,R)=2^{2^{nR}}$. 
%
As mentioned earlier, Ahlswede and Dueck \cite{AD89} needed randomized encoding for their identification-coding scheme. This means that a randomized source is available to the sender. The sender can make his encoding dependent on the output of this source.
Therefore, a randomized-encoder identification (RI) code is defined in a similar manner where the encoder is allowed to select a codeword $U_i$ at random according to some conditional input distribution $Q(x^n|i)$.
The RI capacity 
is denoted by $\mathbb{C}_{RI}(\sG_{\,\text{fast}},L)$. 
\begin{remark}
It can be readily shown that in general, if the capacity in an exponential scale is finite, then it is zero in the double exponential scale. Conversely, if the capacity in a double exponential scale is positive, then the capacity in the exponential scale is $+\infty$. This principle will be further generalized in Lemma~\ref{Lem.MiscScale} to any pair of scales $L_1$ and $L_2$ where $L_2$ is dominated by $L_1$. 
\end{remark}
\begin{remark}
We have mentioned molecular communication (MC) as a motivating application of this study.
Recently, there have been significant advances in MC for complex nano-networks.
The Internet of Things incorporates smart devices, which can be accessed and controlled via the Internet.
The advances in nanotechnology contribute to 
the development of devices in the nanoscale range, referred as nanothings. 
The interconnection of nanothings with the Internet is known as Internet of NanoThings (IoNT)
and is the basis for various future healthcare and military applications \cite{Dress15}.
Nanothings are based on synthesized materials, using electronic circuits, and EM-based communication. Unfortunately, these characteristics could be harmful for some application environments, such
as inside the human body. Furthermore, the concept of Internet of Bio-NanoThings (IoBNT) has been introduced in \cite{Aky15}, where nanothings are biological cells that are created using tools from synthetic biology and nanotechnology. Such biological nanothings are called bio-nanothings.
Similar to artificial nanothings, bio-nanothings have control (cell nucleus), power (mitochondrion), communication (signal pathways), and sensing/actuation (flagella, pili or cilia) units.
For the communication between cells, MC is well suited, since the natural exchange of information between cells is already based on this paradigm. MC in cells is based on signal pathways (chains of chemical reactions) that process information that is modulated into chemical characteristics,
such as molecule concentration. 
Identification is of interest for these applications.
However, it is not clear how randomized identification (RI) codes can be incorporated into such systems. 
In the case of Bio-NanoThings, it is uncertain whether the natural biological processes can be controlled or reinforced by local randomness at this level.
Therefore, for the design of synthetic IoNT, or for the analysis and utilization of IoBNT,  identification with deterministic encoding is more appropriate.
\end{remark}
A geometric illustration for the type I and II error probabilities is given in Figure~\ref{Fig.GeometricID}. When the encoder sends the message $i$, but the channel output is outside $\D_i$, then a type I error occurs. This kind of error is also considered in traditional transmission. In identification, the decoding sets can overlap. A type II error covers the case where the output sequence belongs to the intersection of $\D_i$ and $\D_j$ for $j\neq i$.
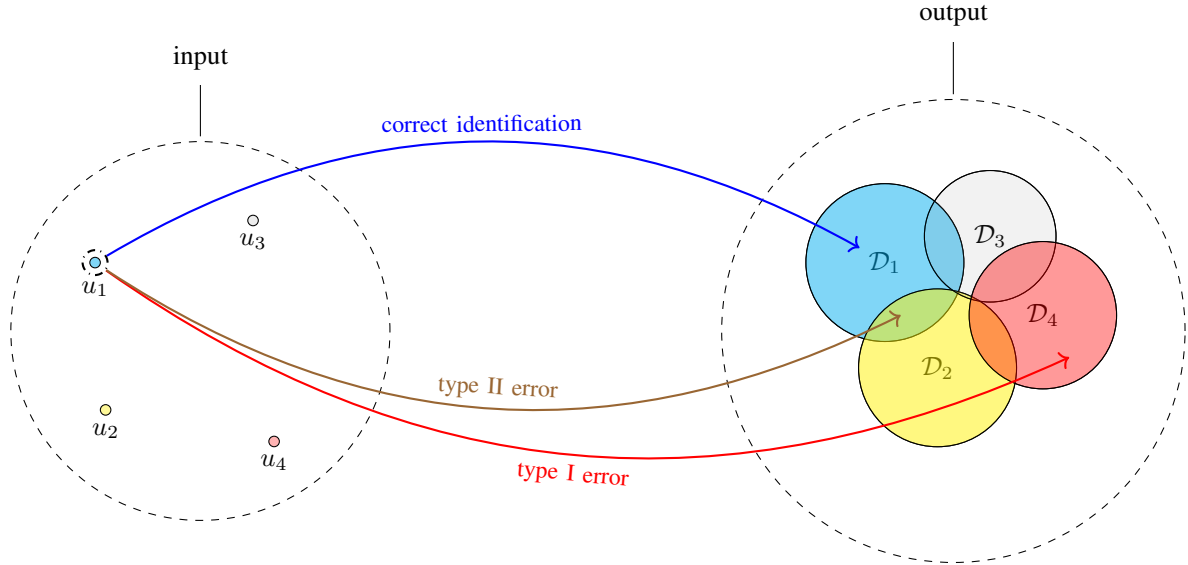
\begin{figure}[ht!]
    \label{Fig.GeometricID}
    \scalebox{1}{
\begin{tikzpicture}[scale=.7][thick]
\draw[dashed] (0,.7) circle (3.6cm);
\draw[->] (0,4.4) -- ++(0,1.5)  node [fill=white,inner sep=6pt](a){$\text{input}$};

\node[circle,draw, dash dot, thick, minimum size=.25mm] (u1) at (-2,2) {};




\node at (-2,1.5) (u1n) {$u_1$};
\node at (-1.8,-1.2) (u2n) {$u_2$};
\node at (1,2.4) (u3n) {$u_3$};
\node at (1.4,-1.8) (u4n) {$u_4$};

\draw [fill=gray!30!white, fill opacity=0.5, name path=i] (1,2.8) circle (.1cm);
\draw [fill=cyan, fill opacity=0.5, name path=k] (-2,2) circle (.1cm);
\draw [fill=red, fill opacity=0.3, name path=k] (1.4,-1.4) circle (.1cm);
\draw [fill=yellow, fill opacity=0.4, name path=k] (-1.8,-.8) circle (.1cm);
\draw[dashed] (14.3,.7) circle (4.4cm);
\draw[->] (14.3,5.2) -- ++(0,1.5)  node [fill=white,inner sep=6pt](a){$\text{output}$};

\draw (13,2) circle (1.5cm);
\draw (14,0) circle (1.5cm);
\draw (15,2.5) circle (1.25cm);
\draw (16,1) circle (1.4cm);

\node at (13,2) (D1n) {$\D_1$};
\node at (14,0) (D2n) {$\D_2$};
\node at (15,2.5) (D3n) {$\D_3$};
\node at (16,1) (D4n) {$\D_4$};

\draw [fill=cyan, fill opacity=0.5, name path=k] (13,2) circle (1.5cm);

\draw [fill=yellow, fill opacity=0.5, name path=k] (14,0) circle (1.5cm);

\draw [fill=gray, fill opacity=0.1, name path=i] (15,2.5) circle (1.25cm);

\draw [fill=red, fill opacity=0.4, name path=j] (16,1) circle (1.4cm);


\path (u1) edge [-> , thick, blue, bend left] node [sloped,midway,above,font=\small] {correct identification}(D1n);

\path (u1) edge [-> , thick, red, bend right] node [sloped,midway,below,font=\small] {type I error}(16.5,0.2);

\path (u1) edge [-> , thick, brown!80!black, bend right] node [sloped,midway,above,font=\small] {type II error}(13.3,1);

\end{tikzpicture}
}
   \caption{Geometric illustration of  identification errors in the deterministic setting. The arrows indicate three scenarios for the channel output, given that the encoder transmitted the codeword $u_1$ corresponding to $i=1$.
   If the channel output is outside $\D_1$, then a type I error has occurred, as indicated by the  bottom red arrow. This kind of error is also considered in traditional transmission. In identification, the decoding sets can overlap. If the channel output belongs to $\D_1$ but also belongs to $\D_2$, then a type II error has occurred, as indicated by the middle brown arrow. Correct identification occurs when the channel output belongs \emph{only} in 
   $\D_1$, which is marked in blue.}
\end{figure}
\subsection{Related Work}
We briefly review known results for the standard Gaussian channel. We begin with 
the RI capacity, i.e., when the encoder uses a stochastic mapping. 
Let $\G$ denote the standard Gaussian channel,
$
    Y_t=gX_t+Z_t
$, 
where the gain $g>0$ is a deterministic constant which is known to the encoder and the decoder. As mentioned above, using RI codes, it is possible to identify a double-exponential number of messages in the block length $n$.
That is, given a rate $R<\mathbb{C}_{RI}(\G,L)$, with $L(n,R)=2^{2^{nR}}$, there exists a sequence of $(2^{2^{nR}},n)$ RI codes with vanishing error probabilities.
Despite the significant difference between the definitions in the identification setting and in the transmission setting, it was shown that the value of the RI capacity in the double-exponential scale equals the Shannon capacity of transmission \cite{AD89,LDB20}.
\begin{theorem}[see \cite{AD89,LDB20}]
 \label{Th.DMCIdentificationCapacity1}
	The RI capacity in the double-exponential scale of the standard Gaussian channel is given by
	\begin{align}
	\mathbb{C}_{RI}(\G,L) =\frac{1}{2}\log\left( 1+\frac{g^2 A}{\sigma_Z^2} \right) \,,\;\text{ for 
	$L(n,R)=2^{2^{nR}}$} \;.\,
	\end{align} 
	Hence, the RI capacity in the exponential scale, i.e., for 
	$L(n,R)=2^{nR}$ is infinite, that is,
	\begin{align}
	\mathbb{C}_{RI}(\G,L)=\infty \;.\,
	\end{align}
\end{theorem}
In a recent paper by the authors, the deterministic case was considered in the exponential scale.
\begin{theorem}[see \cite{SPBD21ICC,SPBD20arXiv}]
\label{Th.DMCIdentificationCapacity2}
	The DI capacity of the standard Gaussian channel in the exponential scale, i.e., for 
	$L(n,R)={2^{nR}}$ is infinite, that is,
	\begin{align}
	\mathbb{C}_{DI}(\G,L) = \infty \;.\,
	\end{align} 
\end{theorem}
Our results in \cite{SPBD21ICC,SPBD20arXiv} reveal a gap of knowledge in the following sense. For a finite blocklength $n$, the number of codewords must be finite. Thereby, the meaning of the result in Theorem~\ref{Th.DMCIdentificationCapacity2} is that the number of messages scales super-exponentially. The question remains what is the true order of the code size.
In mathematical terms, what is the scale $L$ for which the DI capacity is positive yet finite. In the next section, we provide an answer to this question.
\section{Main Results - Channels with Fast Fading}
\label{Sec.GaussianChannelFast}
Before we give our results, we state the following property of dominating coding scales, as defined in Definition~\ref{Def.DominatingScale}. Recall that we use the notation of $L_2\prec L_1$ for a coding scale $L_1$ that dominates $L_2$. The following Lemma readily follows from the Definition~\ref{Def.DominatingScale}.
\begin{lemma}
\label{Lem.MiscScale}
Suppose that the capacity in $L_0$-scale is positive yet finite, i.e., $0<\mathbb{C}_{DI}(\sG_{\,\text{fast}}, L_0)<\infty$. Then, for every $L^{-} \prec L_0$,
\begin{align}
 \mathbb{C}_{DI}(\sG_{\,\text{fast}}, L^{-}) =  \infty  \;.\,
\intertext{
and for every $L_0\prec L^{+}$, we have
}
 \mathbb{C}_{DI}(\sG_{\,\text{fast}}, L^{+}) =  0 \;.\,
\end{align}
\end{lemma}
The proof of Lemma~\ref{Lem.MiscScale} is given in Appendix~\ref{App.MiscScale}.

Our DI capacity theorem for the Gaussian channel with fast fading is stated below.
\begin{theorem}
\label{Th.GDICapacityFast}
Assume that the fading coefficients are positive and bounded away from zero, i.e., $0 \notin \text{cl}(\G)$.
The DI capacity of the Gaussian channel $\sG_{\,\text{fast}}$ with fast fading in the $2^{n\log(n)}$-scale, i.e., for $L(n,R)=2^{(n\log n)R}$ is bounded by
\begin{align}
    \label{Eq.GDICapacityFast1}
    \frac{1}{4}\leq\mathbb{C}_{DI}(\sG_{\,\text{fast}},L) \leq 1 \;.\,
\end{align}
Hence, the DI capacity is infinite in the exponential scale  and zero in the double-exponential, i.e.,
\begin{align}
    \label{Eq.GDICapacityFast2}
    \mathbb{C}_{DI}(\sG_{\,\text{fast}},L) = 
    \begin{cases}
    \infty & \text{ for 
	$L(n,R)=2^{nR}$} \;,\,
	\\
	0&\text{ for 
	$L(n,R)=2^{2^{nR}}$} \;.\,
\end{cases}
\end{align}
\end{theorem}
The proofs for the lower and upper bounds in the first part of Theorem~\ref{Th.GDICapacityFast} are given in Section~\ref{Subsec.AchievFast} and Section~\ref{Subsec.ConvFast}, respectively. 
The second part of the theorem is a direct consequence of Lemma~\ref{Lem.MiscScale}.
\begin{remark}
The code scale can also be thought of as a sequence of monotonically increasing functions $L_n(R)$ of the rate. Hence, given a code of size $M=L_n(R)$, the coding rate can be obtained from the inverse relation $R=L_n^{-1}(M)$. In particular, for the transmission setting \cite{S48}, or DI coding for a DMC \cite{J85}, the coding rate is defined as
\begin{align}
    R=\frac{1}{n}\log(M) \;.\,
\end{align}
Whereas for RI coding \cite{AD89}, the rate was defined as
\begin{align}
    R=\frac{1}{n}\log\log(M) \;.\,
\end{align}
On the other hand, using the scale $L(n,R)=2^{n\log(n)R}$ as for Gaussian channels stated in Theorem~\ref{Th.GDICapacityFast} above, the coding rate is
\begin{align}
    R=\frac{\log M}{n\log n} \;.\,
\end{align}
\end{remark}
\subsection{Lower Bound (Achievability Proof)}
\label{Subsec.AchievFast}
Consider the Gaussian channel $\mathscr{G}_{\,\text{fast}}$ with fast fading.
We show that the DI capacity is bounded by  $\mathbb{C}_{DI}(\sG_{\,\text{fast}},L) \geq \frac{1}{4}$ for $L(n,R)=2^{n\log(n)R}$. Achievability is established using a dense packing arrangement and a simple distance-decoder.
A DI code for the Gaussian channel $\sG_{\,\text{fast}}$ with fast fading is constructed as follows.
Consider the normalized input-output relation,
\begin{align}
    \bar{\fY}=\fG\circ\bar{\fx}+\bar{\fZ} \;.\,
    \label{Eq.IOFast}
\end{align}
where the noise sequence $\bar{\fZ}$ is i.i.d. $\sim \mathcal{N}\left(0,\frac{\sigma_Z^2}{n}\right)$, and an input power constraint 
\begin{align}
    \norm{\bar{\fx}} \leq \sqrt{A} \;,\,
\end{align}
with $\bar{\fx} = \frac{1}{\sqrt{n}}\fx$, $\bar{\fZ} = \frac{1}{\sqrt{n}}\fZ$, and $\bar{\fY} = \frac{1}{\sqrt{n}} \fY$. Assuming $0 \notin \text{cl}(\G)$, there exists a positive number $\gamma$ such that
\begin{align}
    \label{Ineq.FadingCoeff}
    |G_t| > \gamma \;,\,
\end{align}
for all $t$ with probability $1$.
\subsubsection*{Codebook construction}
\label{Subsec.CodebookConstructionGaussian}
We use a packing arrangement of non-overlapping hyper-spheres of radius $\sqrt{\epsilon_n}$ in a hyper-sphere of radius $(\sqrt{A}-\sqrt{\epsilon_n})$, with
\begin{align}
    \epsilon_n = \frac{A}{n^{\frac{1}{2}(1-b)}} \;,\,
\end{align}
where $b>0$ is arbitrarily small.
Let $\mathscr{S}$ denote a sphere packing, i.e., an arrangement of $L$ non-overlapping spheres $\S_{\fu_i}(n,r_0)$, $i\in [\![L(n,R)]\!]$ that cover a larger sphere $\S_{\mathbf{0}}(n,r_1)$, with $r_1>r_0$. 
As opposed to standard sphere packing coding techniques, the small spheres are not necessarily entirely contained within the bigger sphere. That is, we only require that the spheres are disjoint from each other and have a non-empty intersection with $\S_{\mathbf{0}}(n,r_1)$. See illustration in  Figure~\ref{Fig.Density}. The packing density $\Delta_n(\mathscr{S})$ is defined as the fraction of the large sphere volume $\text{Vol}\left(\S_{\mathbf{0}}(n,r_1)\right)$ that is  covered by the small spheres, i.e.,
\begin{align}
    \Delta_n(\mathscr{S}) \triangleq \frac{\text{Vol}\left(\S_{\f0}(n,r_1)\cap\bigcup_{i=1}^{L}\S_{\fu_i}(n,r_0)\right)}{\text{Vol}(\S_{\f0}(n,r_1))} \;,\,
    \label{Eq.DensitySphereFast}
\end{align}
(see \cite[Ch.~1]{CHSN13}). A sphere packing is called \emph{saturated} if no spheres can be added to the arrangement without overlap.
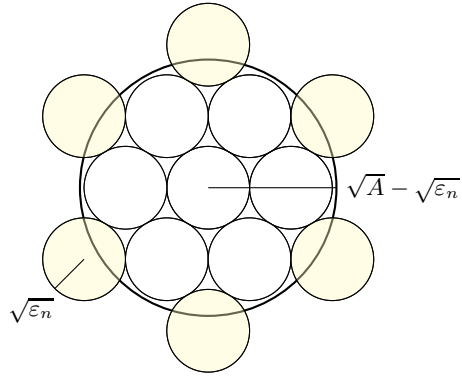
\begin{figure}[htb]
    \centering
	\scalebox{1}{

\begin{tikzpicture}[scale=.55][thick]

\draw[thick] (0,0) circle (3.1cm);

\draw (0,0) circle (1cm);
\draw [] (0,0) circle (1cm);

\draw (2,0) circle (1cm);
\draw [] (2,0) circle (1cm);

\draw (1,1.73) circle (1cm);
\draw [] (1,1.73) circle (1cm);

\draw (-1,1.73) circle (1cm);
\draw [fill=white, fill opacity=0.5] (-1,1.73) circle (1cm);

\draw (-2,0) circle (1cm);
\draw [fill=white, fill opacity=0.5] (-2,0) circle (1cm);

\draw (-1,-1.73) circle (1cm);
\draw [fill=white, fill opacity=0.3] (-1,-1.73) circle (1cm);

\draw (1,-1.73) circle (1cm);
\draw [fill=white, fill opacity=0.4] (1,-1.73) circle (1cm);


\draw (3,-1.73) circle (1cm);
\draw [fill=yellow!30!white, fill opacity=0.4] (3,-1.73) circle (1cm);

\draw (3,1.73) circle (1cm);
\draw [fill=yellow!30!white, fill opacity=0.4] (3,1.73) circle (1cm);

\draw (0,2*1.73) circle (1cm);
\draw [fill=yellow!30!white, fill opacity=0.4] (0,2*1.73) circle (1cm);

\draw (-3,1.73) circle (1cm);
\draw [fill=yellow!30!white, fill opacity=0.4] (-3,1.73) circle (1cm);

\draw (-3,-1.73) circle (1cm);
\draw [fill=yellow!30!white, fill opacity=0.4] (-3,-1.73) circle (1cm);

\draw (0,-2*1.73) circle (1cm);
\draw [fill=yellow!30!white, fill opacity=0.4] (0,-2*1.73) circle (1cm);

\draw (0,0) -- (3.1,0) node [right,font=\small] {$\sqrt{A}-\sqrt{\epsilon_n}$};

\draw (-3,-1.73)-- (-3.707,-2.437)    node [below,font=\small] {$\sqrt{\epsilon_n}\qquad$};

\end{tikzpicture}}
	\caption{Illustration of a sphere packing, where small spheres of radius $r_0=\sqrt{\epsilon_n}$ cover a bigger sphere of radius $r_1=\sqrt{A}-\sqrt{\epsilon_n}$. The small spheres are disjoint from each other and have a non-empty intersection with  the large sphere. Some of the small spheres, marked in yellow, are not  entirely contained within the bigger sphere, and yet they are considered to be a part of the packing arrangement. As we assign a codeword to each small sphere center, the norm of a codeword is bounded by $\sqrt{A}$ as required.}
	\label{Fig.Density}
\end{figure}
We use a packing argument that has a similar flavor as in the Minkowski--Hlawka theorem in lattice theory \cite{CHSN13}. We use the property that there exists an arrangement $\bigcup_{i=1}^{L} \S_{\fu_i}(n,\sqrt{\epsilon_n})$ of non-overlapping spheres inside $\S_{\f0}(n,\sqrt{A})$ with a density of $\Delta_n(\mathscr{S})\geq 2^{-n}$ \cite[Lem.~2.1]{C10}. Specifically, consider a saturated packing arrangement of $L(n,R)=2^{n\log(n)R}$ spheres of radius $r_0=\sqrt{\epsilon_n}$ covering the large sphere $\S_{\f0}(n,r_1=\sqrt{A}-\sqrt{\epsilon_n})$, i.e., such that no spheres can be added without overlap. Then, for such an arrangement, there cannot be a point in the large sphere $\S_{\f0}(n,r_1)$ with a distance of more than $2r_0$ from all sphere centers. Otherwise, a new sphere could be added. As a consequence, if we double the radius of each sphere, the  $2r_0$-radius spheres cover the whole sphere of radius $r_1$. In general, the volume of a hyper-sphere of radius $r$ is given by
\begin{align}
    \text{Vol}\left(\S_\fx(n,r)\right)=\frac{\pi^{\frac{n}{2}}}{\Gamma(\frac{n}{2}+1)} \cdot r^{n} \;.\,
    \label{Eq.VolS}
\end{align}
(\cite[see Eq.~16]{CHSN13}).
Hence, doubling the radius multiplies the volume by $2^n$. Since the $2r_0$-radius spheres cover the entire sphere of radius $r_1$, it follows that the original $r_0$-radius packing has density at least $2^{ -n}$, i.e.,
\begin{align}
    \Delta_n(\mathscr{S})\geq 2^{-n} \;.\,
    \label{Eq.MinkowskiDeltaFast}
\end{align}

We assign a codeword to the center $\fu_i$ of each small sphere.
The codewords satisfy the input constraint as
\begin{align}
    \norm{\fu_i} & \leq r_0+r_1
    \nonumber\\
    & = \sqrt{A} \;.\,
\end{align}
Since the small spheres have the same volume, the total number of spheres is bounded from below by
\begin{align}
    L& =\frac{\text{Vol}\left(\bigcup_{i=1}^{L}\S_{\fu_i}(n,r_0)\right)}{\text{Vol}(\S_{\fu_1}(n,r_0))}
    \nonumber\\&
    \geq\frac{\text{Vol}\left(\S_{\f0}(n,r_1)\cap\bigcup_{i=1}^{L}\S_{\fu_i}(n,r_0)\right)}{\text{Vol}(\S_{\fu_1}(n,r_0))}
    \nonumber\\&
    =\frac{\Delta_n(\mathscr{S})\cdot
    \text{Vol}(\S_{\mathbf{0}}(n,r_1)))}{\text{Vol}(\S_{\fu_1}(n,r_0))}
    \nonumber\\&
    \geq 2^{-n}\cdot \frac{
    \text{Vol}(\S_{\mathbf{0}}(n,r_1)))}{\text{Vol}(\S_{\fu_1}(n,r_0))}
    \nonumber\\
    & = 2^{-n}\cdot \frac{r_1^n}{r_0^n} \;,\,
\end{align}
where the second equality is due to (\ref{Eq.DensitySphereFast}),  the inequality that follows holds by (\ref{Eq.MinkowskiDeltaFast}), and the last equality follows from (\ref{Eq.VolS}).
That is, the codebook size satisfies
\begin{align}
    L(n,R) & = 2^{n\log(n)R}
    \nonumber\\
    & \geq 2^{-n} \cdot \left(\frac{\sqrt{A}-\sqrt{\epsilon_n}}{\sqrt{\epsilon_n}}\right)^n \;.\,
\end{align}
Hence,
\begin{align}
    \label{Eq.RateFast}
    R&\geq \frac{1}{\log(n)}\log\left(\frac{\sqrt{A}-\sqrt{\epsilon_n}}{\sqrt{\epsilon_n}}\right)-\frac{1}{\log(n)}\nonumber\\&
    =\frac{1}{\log(n)}\log\left( n^{\frac{1}{4}(1-b)}-1 \right)
    -\frac{1}{\log(n)}
    \nonumber\\&
    \geq \frac{1}{\log(n)}\left(\log n^{\frac{1}{4}(1-b)} - 1 \right)
    -\frac{1}{\log(n)}
    \nonumber\\& 
    =\frac{1}{4}(1-b)-\frac{2}{\log(n)} \;,\,
\end{align}
which tends to $\frac{1}{4}$ when $n\to\infty$ and $b\to 0$, where the second inequality holds since $\log(t-1)\geq \log(t)- 1$ for $t\geq 2$.
\subsubsection*{Encoding}
Given a message $i\in [\![L(n,R)]\!]$, 
 transmit $\bar{\fx}=\bar{\fu}_i$.

\subsubsection*{Decoding}
Let
\begin{align}
    \label{Eq.beta_n}
    \delta_n & = \frac{\gamma^2\epsilon_n}{3}
    \nonumber\\
    & = \frac{\gamma^2 A}{3n^{\frac{1}{2}(1-b)}} \;.\,
\end{align}
To identify whether a message $j\in [\![L(n,R)]\!]$ was sent, given the sequence $\fg$, the decoder checks whether the channel output $\bar{\fy}$ belongs to the following decoding set,
\begin{align}
    \D_{j,\fg} = \left\{ \bar{\fy} \in \mathbb{R}^n \,:\; \norm{\bar{\fy}-\fg\circ\bar{\fu}_j} \leq \sqrt{\sigma_Z^2+\delta_n} \right\}  \;.\,
\end{align}
\subsubsection*{Error Analysis}
Consider the type I error, i.e., when the transmitter sends $\bar{\fu}_i$, yet $\bar{\fY}\notin\D_{i,\fG}$. For every $i\in[\![L(n,R)]\!]$, the type I error probability is bounded by
\begin{align}
    P_{e,1}(i)&= \Pr\left(\norm{\bar{\fY}-\fG\circ\bar{\fu}_i}^2 > \sigma_Z^2+\delta_n \,\big|\,\bar{\fx} = \bar{\fu}_i \right)
    \nonumber\\
    &=\Pr\left(\norm{\bar{\fZ}}^2> \sigma_Z^2+\delta_n  \right)
    \nonumber\\
    &=\Pr\left(\sum_{t=1}^n {\bar{Z}_t}^2 > \sigma_Z^2+\delta_n \right)
    \nonumber\\
    &\leq \Pr\left(\sum_{t=1}^n {\bar{Z}_t}^2> \sigma_Z^2+\delta_n \right)
    \nonumber\\
    & \leq \frac{3\sigma_Z^4}{n\delta_n^2}
    \nonumber\\
    & = \frac{27\sigma_Z^4}{n^bA^2\gamma^4}
    \nonumber\\
    & \leq \lambda_1 \;,\,
\end{align}
for sufficiently large $n$ and arbitrarily small $\lambda_1>0$, where the second inequality follows by Chebyshev's inequality, and since the fourth moment of a Gaussian variable $V\sim \N(0,\sigma_V^2)$ is $\mathbb{E}\{V^4\}=3\sigma_V^4$.

Next, we address the type II error, i.e., when $\bar{\fY}\in\D_{j,\fG}$ while the transmitter sent $\bar{\fu}_i$.
Then, for every $i,j\in[\![L(n,R)]\!]$, where $i\neq j$, the type II error probability is given by
\begin{align}
    P_{e,2}(i,j)&= \Pr\left( \norm{\bar{\fY}-\fG\circ\bar{\fu}_j}^2\leq \sigma_Z^2+\delta_n \,\big|\,\bar{\fx}=\bar{\fu}_i \right)
    \nonumber\\
    &=\Pr\left( \norm{\fG\circ(\bar{\fu}_i-\bar{\fu}_j)+\bar{\fZ}}^2\leq \sigma_Z^2+\delta_n \right) \;.\,
    \label{Eq.Pe2GFast}
\end{align}
Observe that the square norm can be expressed as
\begin{align}
    \norm{\fG\circ(\bar{\fu}_i-\bar{\fu}_j)+\bar{\fZ}}^2=
    \norm{\fG\circ(\bar{\fu}_i-\bar{\fu}_j)}^2+\norm{\bar{\fZ}}^2+2\sum_{t=1}^n G_t(\bar{u}_{i,t}-\bar{u}_{j,t})\bar{Z}_t \;.\,
     \label{Eq.Pe2normFast}
\end{align}
Then, define the event
\begin{align}
    \E_0 = \left\{ \left| \sum_{t=1}^n G_t(\bar{u}_{i,t} - \bar{u}_{j,t})\bar{Z}_t \right| > \frac{\delta_n}{2} \right\} \;.\,
    \label{Eq.E0FadingFast}
\end{align}
By Chebyshev's inequality, the probability of this event vanishes, 
\begin{align}
    \Pr(\E_0)&\leq 
     \frac{\sum_{t=1}^n (\bar{u}_{i,t}-\bar{u}_{j,t})^2\mathbb{E}\{G_t^2\}\mathbb{E}\{\bar{Z}_t^2\}}{\left(\frac{\delta_n}{2} \right)^2}
    \nonumber\\
    & =\frac{\sigma^2_Z (\sigma^2_G+\mu_G^2)\sum_{t=1}^n (\bar{u}_{i,t}-\bar{u}_{j,t})^2}{n\left(\frac{\delta_n}{2} \right)^2}
    \nonumber\\
    & =
    \frac{4\sigma^2_Z (\sigma^2_G+\mu_G^2)\norm{\bar{\fu}_i-\bar{\fu}_j}^2}{n\delta_n^2} \;,\,
    \label{Eq.PeE0G1Fast}
\end{align}
where the first inequality holds since the sequences $\{\bar{Z}_t\}$ and $\{G_t\}$ are i.i.d. $\sim\mathcal{N}\left(0,\frac{\sigma_Z^2}{n}\right)$ and $\sim f_{G}$ with
\begin{align}
   \mathbb{E}\{G_t\}=\mu_G \quad \text{and} \quad \mathbb{E}\{G_t^2\}=\sigma_G^2+\mu_G^2 \;.\,
\end{align}
By the triangle inequality,
\begin{align}
    \norm{\bar{\fu}_i-\bar{\fu}_j}^2 & \leq (\norm{\bar{\fu}_i} + \norm{\bar{\fu}_j})^2
    \nonumber\\
    & \leq \left( \sqrt{A} + \sqrt{A} \right)^2
    \nonumber\\
    & = 4A \;,\,
 \end{align}
 hence
\begin{align}
    \Pr(\E_0) & \leq
    \frac{16A\sigma^2_Z (\sigma_G^2+\mu_G^2)}{n\delta_n^2}
    \nonumber\\
    & = \frac{144\sigma^2_Z (\sigma_G^2+\mu_G^2)}{\gamma^4 A n^b}
    \nonumber\\
    & \leq \eta_1 \;,\,
    \label{Eq.PeE0G}
\end{align}
for sufficiently large $n$, with arbitrarily small $\eta_1>0$. Furthermore, observe that given the complementary event $\E_0^c$, we have 
\begin{align}
    2\sum_{t=1}^n G_t(\bar{u}_{i,t}-\bar{u}_{j,t})\bar{Z}_t\geq -\delta_n \;,\,    
\end{align}
Therefore, the event $\E_0^c$, the type II error event in (\ref{Eq.Pe2GFast}), and the identity in (\ref{Eq.Pe2normFast}) together imply that the following event occurs,
\begin{align}
    \E_1 & = \left\{ \norm{\fG\circ(\bar{\fu}_i-\bar{\fu}_j)}^2+\norm{\bar{\fZ}}^2\leq \sigma_Z^2+2\delta_n \right\} \;.\,
    \label{Eq.Pe2normConsequenceFast}
\end{align}
Now lets define
\begin{align}
    \G_{i,j}^n = \left\{ \fG \in \G^n \;:\, \norm{\fG\circ(\bar{\fu}_i-\bar{\fu}_j)+\bar{\fZ}}^2 \leq \sigma_Z^2 + \delta_n \right\} \;.\,
\end{align}
Therefore, applying the law of total probability, we obtain
\begin{align}
    P_{e,2}(i,j) & \stackrel{(a)}{=} \Pr\left( \G_{i,j}^n \; \cap \; \E_0 \right) + \Pr\left( \G_{i,j}^n \; \cap \; \E_0^c \right)
    \nonumber\\
    & \stackrel{(b)}{\leq} \eta_1 + \Pr(\E_1) \;,\,
    \label{Eq.Pe2_P_Expanded_L}
\end{align}
where $(a)$ is due to (\ref{Eq.Pe2GFast}) and $(b)$ holds since each probability is bounded by 1.

Based on the codebook construction, each codeword is surrounded by a sphere of radius $\sqrt{\epsilon_n}$, which implies that
\begin{align}
    \norm{\bar{\fu}_i-\bar{\fu}_j} \geq \sqrt{\epsilon_n} \;.\,
\end{align}
Then, by (\ref{Ineq.FadingCoeff}),
\begin{align}
    \norm{\fG\circ(\bar{\fu}_i-\bar{\fu}_j)}^2 & \geq \gamma^2 \norm{\bar{\fu}_i-\bar{\fu}_j}^2
    \nonumber\\
    & \geq \gamma^2\epsilon_n \;,\,
\end{align}
where
$\gamma$ is the minimal value in ${\G}$.
Hence, according to (\ref{Eq.Pe2_P_Expanded_L}),
\begin{align}
    P_{e,2}(i,j) 
    & \leq \Pr\left( \norm{\bar{\fZ}}^2\leq \sigma_Z^2+2\delta_n-\gamma^2\epsilon_n
    \right) + \eta_1
    \nonumber\\
    & \leq \Pr\left( \norm{\bar{\fZ}}^2\leq \sigma_Z^2 -\delta_n
    \right) + \eta_1 \;,\,
\end{align}
where the last line holds, since $2\delta_n-\gamma^2\epsilon_n=-\delta_n$ by (\ref{Eq.beta_n}).
Therefore, by Chebyshev's inequality,
\begin{align}
    P_{e,2}(i,j) & \leq \Pr\left(\sum_{t=1}^n \bar{Z}_t^2-\sigma_Z^2 \leq -\delta_n \right)  + \eta_1
    \nonumber\\
    & \leq \frac{\sum_{t=1}^n\text{var}(\bar{Z}_t^2)}{\delta_n^2}+ \eta_1
    \nonumber\\
    & \leq \frac{\sum_{t=1}^n \mathbb{E}\{\bar{Z}_t^4\}}{\delta_n^2}+ \eta_1
    \nonumber\\
    & = \frac{ 3n \left(\frac{\sigma_Z^2}{n}\right)^2}{\delta_n^2} + \eta_1
    \nonumber\\
    & = \frac{27\sigma_Z^4}{\gamma^4 A^2n^b} + \eta_1
    \nonumber\\
    & \leq \lambda_2 \;,\,
\end{align}
for sufficiently large $n$, where $\lambda_2$ is arbitrarily small.

We have thus shown that for every $\lambda_1,\lambda_2>0$ and sufficiently large $n$, there exists a $(2^{n\log (n)R},n,\lambda_1,\lambda_2)$ code.
As we take the limits of $n\rightarrow\infty$, and then $b \rightarrow 0$, the lower bound on the achievable rate tends to $\frac{1}{4}$, by (\ref{Eq.RateFast}). This completes the achievability proof for Theorem~\ref{Th.GDICapacityFast}.
\qed
\subsection{Upper Bound (Converse Proof)}
\label{Subsec.ConvFast}
We show that the capacity is bounded by $\mathbb{C}_{DI}(\sG_{\text{fast}},L)\leq 1$. We note that in the converse proof, we do \emph{not} normalize the sequences. Suppose that $R$ is an achievable rate in the $L$-scale for the Gaussian channel with fast fading. Consider a sequence of $(L(n,R),n,\lambda_1^{(n)},\lambda_2^{(n)})$ codes $(\U^{(n)},\D^{(n)})$ such that $\lambda_1^{(n)}$ and $\lambda_2^{(n)}$ tend to zero as $n\rightarrow\infty$. We begin with the following lemma.
\begin{lemma}
\label{Lem.DConverseFast}
Consider a sequence of codes  as described above.
Let $b>0$ be an arbitrarily small constant that does not depend on $n$.
%
 Then there exists $n_0(b)$, such that for all $n>n_0(b)$, every pair of codewords in the codebook $\U^{(n)}$ are distanced by at least $\sqrt{n\epsilon_n}$, i.e.,
    \begin{subequations}
    \begin{align}
     \norm{\fu_{i_1} - \fu_{i_2}}\geq \sqrt{n\epsilon_n} \;,\,
    \end{align}
    where
    \begin{align}
        \epsilon_n = \frac{A}{n^{2(1+b)}} \;,\,
    \end{align}
    \end{subequations}
for all $i_1,i_2\in [\![L(n,R)]\!]$ such that $i_1\neq i_2$.
\end{lemma}
\begin{proof}
    Fix $\lambda_1$ and $\lambda_2$. Let $\kappa,\theta,\zeta>0$ be arbitrarily small.
    Assume to the contrary that 
    there exist two messages $i_1$ and $i_2$, where $i_1\neq i_2$, such that
    \begin{align}
        \label{Eq.Alpha_nFast}
        \norm{\fu_{i_1} - \fu_{i_2}}< \sqrt{n\epsilon_n} = \alpha_n \;,\,
    \end{align}
    where
    \begin{align}
     \alpha_n = \frac{\sqrt{A}}{n^{\frac{1}{2}(1+2b)}} \;.\,
    \end{align}
    %

   Observe that 
   \begin{align}
       \mathbb{E}\left\{ \norm{\fG\circ(\fu_{i_1} - \fu_{i_2})}^2 \right\} & = \sum_{t=1}^n \mathbb{E}\left\{G_t^2\right\} \left( \fu_{i_1,t} - \fu_{i_2,t}\right)^2
       \nonumber\\
       & = \mathbb{E}\left\{G^2\right\} \norm{\fu_{i_1} - \fu_{i_2}}^2 \;,\,
   \end{align}
  and consider the subset
   \begin{align}
        \label{Eq.Ai1i2}
           \A_{i_1,i_2}=\{ \fg\in\G^n \,:\; \norm{\fg\circ(\fu_{i_1} - \fu_{i_2})}> \delta_n  \} \;,\,
   \end{align}
   where
   \begin{align}
        \label{Eq.Deltan_Fast}
        \delta_n = \frac{\sqrt{A}}{n^{\frac{1}{2}(1+b)}} \;.\,
    \end{align}
   By Markov's inequality, the probability that the fading sequence $\fG$ belongs to this set is bounded by 
   \begin{align}
       \label{Ineq.A_Event_Fast}
       \Pr\left(\fG\in \A_{i_1,i_2}\right) & = \Pr(\norm{\fG\circ(\fu_{i_1} - \fu_{i_2})}^2 > \delta_n^2)
       \nonumber\\
       & \stackrel{(a)}{\leq} 
       \frac{\mathbb{E}\{G^2\} \norm{\fu_{i_1} - \fu_{i_2}}^2}{\delta_n^2} \nonumber\\
       & \stackrel{(b)}{\leq}  \frac{\mathbb{E}\{G^2\} \alpha_n^2}{\delta_n^2}
       \nonumber\\
       & = \frac{\mathbb{E}\{G^2\}}{n^{b}}
       \nonumber\\
        & \leq \kappa \;,\,
       %
       %
   \end{align}
   for sufficiently large $n$ where $(a)$ holds since the sequence $\left\{G_t\right\}_{t=1}^n$ is i.i.d. and $(b)$ is due to (\ref{Eq.Alpha_nFast}).
   
    Then, observe that 
    \begin{align}
        1-P_{e,1}(i_1)&= \int_{\G^n} f_{\fG}(\fg)\bigg[ \int_{\D_{i_1,\fg}} f_{\fZ}(\fy-\fg\circ\fu_{i_1}) d\fy  \bigg] d\fg
        \nonumber\\
        &\leq
        \int_{\A_{i_1,i_2}^c}  f_{\fG}(\fg)\bigg[ \int_{\D_{i_1,\fg}} f_{\fZ}(\fy-\fg\circ\fu_{i_1}) d\fy  \bigg] d\fg+\Pr(\fG\in\A_{i_1,i_2})
        \nonumber\\
        &\leq
        \int_{\A_{i_1,i_2}^c}  f_{\fG}(\fg)\bigg[ \int_{\D_{i_1,\fg}} f_{\fZ}(\fy-\fg\circ\fu_{i_1}) d\fy  \bigg] d\fg+\kappa
    \end{align}
Now let us define two events as follows
    \begin{align}
        \B_{i_1,i_2} & = \left\{ \fy \in \D_{i_1,\fg} \,:\; 
        \norm{ \fy - \fg \circ \fu_{i,2} } \leq \sqrt{ n \left( \sigma_Z^2 + \zeta \right)} \right\} \;,\,
        \label{Eq.Event_B_Fast}
    \end{align}
    \begin{align}
        \label{Eq.Event_C_Fast}
        \C_{i_1,i_2} & = \left\{ \fy \in \Y^n \,:\;
        \norm{\fy-\fg\circ\fu_{i,2}} \leq \sqrt{n(\sigma_Z^2+\zeta)} \right\} \;.\,
    \end{align}
    Hence,
    \begin{align}
        1-\kappa-P_{e,1}(i_1) & \leq \int_{\A_{i_1,i_2}^c} f_{\fG}(\fg)\left[ \int_{\D_{i_1,\fg}} f_{\fZ}(\fy-\fg\circ\fu_{i_1}) d\fy \right] d\fg
        \nonumber\\
        & = \int_{\A_{i_1,i_2}^c} f_{\fG}(\fg) \left[ \int_{\B_{i_1,i_2}} f_{\fZ}(\fy-\fg\circ\fu_{i_1}) d\fy + \int_{\D_{i_1,\fg} \setminus \B_{i_1,i_2}} f_{\fZ}(\fy-\fg\circ\fu_{i_1}) d\fy \right] d\fg
        \nonumber\\
        & \leq \int_{\A_{i_1,i_2}^c} f_{\fG}(\fg) \left[ \int_{\B_{i_1,i_2}} f_{\fZ}(\fy-\fg\circ\fu_{i_1}) d\fy + \int_{\C_{i_1,i_2}^c} f_{\fZ}(\fy-\fg\circ\fu_{i_1}) d\fy \right] d\fg \;.\,
        \label{Eq.Pe1boundConv0Fast}
        \end{align}
        where the last inequality holds since
        \begin{align}
            \C_{i_1,i_2}^c \supset \D_{i_1,\fg} \setminus \B_{i_1,i_2} \;,\,
        \end{align}
        with $\setminus$ being the set minus operation. Consider the second integral, where the domain is $\C_{i_1,i_2}$. Then, by the triangle inequality
        \begin{align}
           \norm{\fy-\fg\circ\fu_{i,1}}&\geq \norm{\fy-\fg\circ\fu_{i,2}}-\norm{\fg\circ(\fu_{i,1}-\fu_{i,2})}
            \nonumber\\
            &> \sqrt{n(\sigma_Z^2+\zeta)}-\norm{\fg\circ(\fu_{i,1}-\fu_{i,2})}
            \nonumber\\ &\geq\sqrt{n(\sigma_Z^2+\zeta)}-\delta_n \;,\,
        \end{align}
        for every $\fg\in\A_{i_1,i_2}^c$ (see (\ref{Eq.Alpha_nFast})).
        For sufficiently large $n$, this implies the following email
        \begin{align}
           \F_{i_1,i_2}^c = \left\{y^n \in \Y^n \; : \, \norm{\fy - \fg \circ\fu_{i,1}} > \sqrt{n\left( \sigma_Z^2 + \eta \right)} \right\} \;,\,
            \label{Eq.Regiong0}
        \end{align}
        for $\eta<\frac{\zeta}{2}$. 
        That is,
        \begin{align}
            \left\{\fy \in \Y^n \; : \, \norm{\fy-\fg\circ\fu_{i,2}} \geq
        \sqrt{n(\sigma_Z^2+\zeta)} \right\} \quad \overset{\text{implies}}{\longrightarrow} \quad \left\{\fy \in \Y^n \; : \, \norm{\fy-\fg\circ\fu_{i,1}} \geq
        \sqrt{n(\sigma_Z^2+\eta)} \right\} \;.\,
        \end{align}
        Thus, we deduce that for every $\fg\in\A_{i_1,i_2}^c$, 
        \begin{align}
            \F_{i_1,i_2}^c \supset \C_{i_1,i_2}^c \;,\,
        \end{align}
        Hence, the second integral in the right hand side of (\ref{Eq.Pe1boundConv0Fast}) is bounded by
        \begin{align}
            \int_{\F_{i_1,i_2}^c} f_{\fZ}\left( \fy - \fg \circ \fu_{i_1} \right) d\fy & = \Pr \left( \norm{\fy - \fg \circ\fu_{i,1}} > \sqrt{n\left( \sigma_Z^2 + \eta \right)} \right)
            \nonumber\\
            & = \Pr\left( \norm{\fZ}^2 - n\sigma_Z^2 > n\eta \right)
            \nonumber\\
            & = \Pr\left( \norm{\fZ}^2 - n\sigma_Z^2 > n\eta \right)
            \nonumber\\
            & \leq \frac{3\sigma_Z^4}{n\eta^2}
            \nonumber\\
            & \leq \kappa \;,\,
        \end{align}
        for large $n$,  where the third line is due to Chebyshev's inequality, followed by the substitution of
        $\fz\equiv \fy-\fg\circ\fu_{i_1} $.
   Thus, by (\ref{Eq.Pe1boundConv0Fast}),
    \begin{align}
    \label{Eq.ComplTypeIFast}
     1-2\kappa-P_{e,1}(i_1) \leq\int_{\A_{i_1,i_2}^c} f_{\fG}(\fg)\left[ \int_{\B_{i_1,i_2}} f_{\fZ}(\fy-\fg\circ\fu_{i_1}) d\fy \right] d\fg \;.\,
    \end{align}
    
    Now, we can focus on the inner integral with domain of $\B_{i_1,i_2}$, i.e., when
    \begin{align}
        \norm{\fy-\fg\circ\fu_{i,2}}\leq\sqrt{n(\sigma_Z^2+\zeta)} \;.\,
        \label{Eq.ui2DistFast}
    \end{align}
    Observe that
    \begin{align}
        f_{\fZ}(\fy-\fg\circ\fu_{i_1}) - f_{\fZ}(\fy-\fg\circ\fu_{i_2})=  f_{\fZ}(\fy-\fg\circ\fu_{i_1})\left[1-e^{-\frac{1}{2\sigma_Z^2}\left(\norm{\fy-\fg\circ\fu_{i_2}}^2-\norm{\fy-\fg\circ\fu_{i_1}}^2\right)}\right] \;.\,
    \end{align}
    By the triangle inequality,
    \begin{align}
        \label{Ineq.triangleFast}
        \norm{\fy-\fg\circ\fu_{i_1}}\leq  \norm{\fy-\fg\circ\fu_{i_2}} +  \norm{\fg\circ(\fu_{i_1} - \fu_{i_2})} \;.\,
    \end{align}
    Taking the square of both sides, we have
    \begin{align}
        \norm{\fy-\fg\circ\fu_{i_1}}^2 &\leq \norm{\fy-\fg\circ\fu_{i_2}}^2 +  \norm{\fg\circ(\fu_{i_2} - \fu_{i_1})}^2 +
        2\norm{\fy-\fg\circ\fu_{i_2}}\cdot \norm{\fg\circ(\fu_{i_2} - \fu_{i_1})}
        \nonumber\\
        &\leq
        \norm{\fy-\fg\circ\fu_{i_2}}^2+ \delta_n^2+ 2\delta_n\sqrt{n(\sigma_Z^2+\zeta)}
        \nonumber\\
        &=
       \norm{\fy-\fg\circ\fu_{i_2}}^2+ \delta_n^2+
       \frac{2\sqrt{A(\sigma_Z^2+\zeta)}}{n^{\frac{b}{2}}} \;,\,
    \end{align}
    where the last inequality follows from the definition of $\A_{i_1,i_2}$ and $\B_{i_1,i_2}$ according to (\ref{Eq.Ai1i2})
    and (\ref{Eq.Event_B_Fast}), respectively.
    Thus, for sufficiently large $n$,
    \begin{align}
        \norm{\fy-\fg\circ\fu_{i_1}}^2-\norm{\fy-\fg\circ\fu_{i_2}}^2\leq \theta \;.\,
    \end{align}
    Hence,
    \begin{align}
        \label{Ineq.GaussianContinuityFast}
        f_{\fZ}(\fy-\fg\circ\fu_{i_1}) - f_{\fZ}(\fy-\fg\circ\fu_{i_2})&\leq  f_{\fZ}(\fy-\fg\circ\fu_{i_1})(1-e^{-\frac{\theta}{2\sigma_Z^2}})
        \nonumber\\
        &\leq \kappa f_{\fZ}(\fy-\fg\circ\fu_{i_1}) \;,\,
    \end{align}
    for sufficiently small $\theta>0$ such that $1-e^{-\frac{\theta}{2\sigma_Z^2}}\leq \kappa$.
 Now by (\ref{Eq.ComplTypeIFast}) we get
    \begin{align}
       \lambda_1+\lambda_2 &\geq 
       P_{e,1}(i_1) + P_{e,2}(i_2,i_1) 
       \nonumber\\
       &\geq 1-2\kappa - \int_{\A_{i_1,i_2}^c}  f_{\fG}(\fg)\left[
       \int_{\B_{i_1,i_2}}    f_{\fZ}(\fy-\fg\circ\fu_{i_1}) d\fy 
       \right]d\fg + \int_{\G^n}  f_{\fG}(\fg)\left[ \int_{\D_{i_1,\fg}} f_{\fZ}(\fy-\fg\circ\fu_{i_2}) \, d\fy \right] d\fg
       \nonumber\\
       &\geq 1- 2\kappa -\int_{\A_{i_1,i_2}^c}  f_{\fG}(\fg)\left[
       \int_{\B_{i_1,i_2}}    f_{\fZ}(\fy-\fg\circ\fu_{i_1}) d\fy
       \right]d\fg + \int_{\A_{i_1,i_2}^c} \hspace{-2.5mm} f_{\fG}(\fg)\bigg[ \int_{\B_{i_1,i_2}} f_{\fZ}(\fy-\fg\circ\fu_{i_2}) \, d\fy \bigg] d\fg
       \nonumber\\
        &\geq 1- 2\kappa -\int_{\A_{i_1,i_2}^c}  f_{\fG}(\fg)\bigg[
       \int_{\B_{i_1,i_2}}    (f_{\fZ}(\fy-\fg\circ\fu_{i_1})-f_{\fZ}(\fy-\fg\circ\fu_{i_2})) d\fy
       \bigg]d\fg \;.\,
    \end{align}
    Hence, by (\ref{Ineq.GaussianContinuityFast}),
    \begin{align}
        \lambda_1 + \lambda_2 
         &\geq 1- 2\kappa -\kappa\int_{\A_{i_1,i_2}^c}  f_{\fG}(\fg)\bigg[
       \int_{\B_{i_1,i_2}}    f_{\fZ}(\fy-\fg\circ\fu_{i_1}) d\fy
       \bigg]d\fg
       \nonumber\\
       & \geq 1-3\kappa \;,\,
    \end{align}
    which leads to a contradiction for sufficiently small  $\kappa$ such that $3\kappa<1-\lambda_1-\lambda_2$. This completes the proof of Lemma~\ref{Lem.DConverseFast}.
    \end{proof}
By Lemma~\ref{Lem.DConverseFast}, we can define an arrangement of non-overlapping spheres $\S_{\fu_i}(n,\sqrt{n\epsilon_n})$ of radius $\sqrt{n\epsilon_n}$ centered at the codewords $\fu_i$.
Since the codewords all belong to a sphere $\S_{\f0}(n,\sqrt{nA})$ of radius $\sqrt{nA}$ centered at the origin, it follows that the number of packed spheres, i.e., the number of codewords $2^{n\log(n)R}$, is bounded by
\begin{align}
    2^{n\log(n)R} & \leq \frac{\text{Vol}(\S_{\f0}(n,\sqrt{nA}+\sqrt{n\epsilon_n}))}{\text{Vol}(\S_{\fu_i}(n,\sqrt{n\epsilon_n}))}
    \nonumber\\
    & = \left(\frac{\sqrt{A}+\sqrt{\epsilon_n}}{\sqrt{\epsilon_n}}\right)^n \;.\,
\end{align}
Thus,
\begin{align}
    \label{Ineq.Rate_Fast}
     R &\leq \frac{1}{\log n} \log\left(\frac{\sqrt{A}+\sqrt{\epsilon_n}}{\sqrt{\epsilon_n}}\right)
     \nonumber\\
     & = \frac{\log\left(1+n^{1+b}\right)}{\log n}
     \nonumber\\
     & = \frac{\log \left( n^{1+b} \left( 1 + \frac{1}{n^{1+b}} \right) \right)}{\log n}
    \nonumber\\
    & = \frac{(1 + b) \log n + \log \left( 1 + \frac{1}{n^{1 + b}} \right) }{\log n}
    \nonumber\\
    & = 1 + b + \frac{\log \left( 1 + \frac{1}{n^{1 + b}} \right) }{\log n}
     \;,\,
\end{align}
which tends to $1+b$ as $n \to \infty$. Therefore, $R \leq 1 + b$. Now, since $b>0$ is arbitrarily small, an achievable rate must satisfy $R \leq 1$. This completes the proof of Theorem~\ref{Th.GDICapacityFast}.
\qed
\section{The Gaussian Channel with Slow Fading}
\label{Sec.GaussianChannelSlow}

In this section, we consider the Gaussian channel $\sG_{\,\text{slow}}$ with slow fading, specified by the input-output relation
\begin{align}
    Y_t = Gx_t + Z_t \;,\,
\end{align}
where $G$ is a continuous random variable $ \sim f_G(g)$. Suppose that the values of $G$ belong to a set $\G$, and that $G$ has finite expectation, and finite variance  $\text{var}(G)>0$. 
with additive white Gaussian noise, i.e.,
where the noise sequence $\fZ$ is i.i.d. $\sim \mathcal{N}(0,\sigma_Z^2)$. The transmission power is limited to 
$\norm{\textbf{x}}^2\leq nA$.
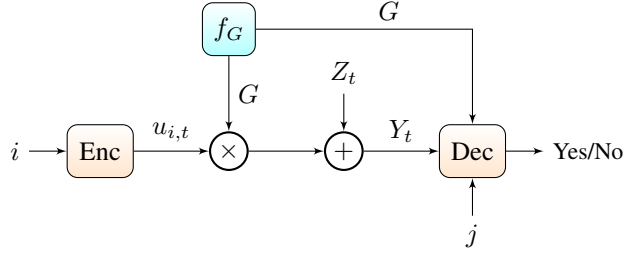
\begin{figure}[htb]
    \centering
	\tikzstyle{l} = [draw, -latex']

\tikzstyle{farbverlauf} = [ top color=white, bottom color=white!80!gray]
\tikzstyle{block1} = [draw,top color=white, bottom color=orange!20!white, rectangle, rounded corners,
minimum height=2em, minimum width=2.5em]

\tikzstyle{block2} = [draw,top color=white, bottom color=white!80!orange, rectangle, rounded corners,
minimum height=2em, minimum width=2.5em]

\tikzstyle{block3} = [draw,top color=white, bottom color=cyan!30, rectangle, rounded corners,
minimum height=2em, minimum width=2em]

\tikzstyle{input} = [coordinate]
\tikzstyle{sum} = [draw, circle,inner sep=0pt, minimum size=5mm,  thick]
\tikzstyle{multiplexer} = [draw, circle,inner sep=0pt, minimum size=5mm,  thick]
\tikzstyle{arrow}=[draw,->]
\begin{tikzpicture}[auto, node distance=2cm,>=latex']
\node[] (M) {$i$};
\node[block1,right=.5cm of M] (enc) {Enc};
\node[multiplexer, right=1cm of enc] (multiplexer) {$\times$};
\node[sum, right=1cm of multiplexer] (channel) {$+$};

\node[block2, right=1cm of channel] (dec) {Dec};
\node[below=.5cm of dec] (Target) {$j$};
\node[right=.5cm of dec] (Output) {$\text{\small Yes/No}$};

\node[block3,above=1cm of multiplexer] (gain) {$f_G$};
\node[above=.5cm of channel] (noise) {$Z_t$};

\path [l] (gain) -| node [text width=2.5cm,above] {$G$} (dec) ;

\draw[->] (M) -- (enc);
\draw[->] (enc) --node[above]{$u_{i,t}$} (multiplexer);
\draw[->] (multiplexer) -- (channel);
\draw[->] (gain) -- node [] {$G$} (multiplexer);
\draw[->] (noise) -- (channel);
\draw[->] (channel) --node[above]{$Y_t$} (dec);

\draw[->] (dec) -- (Output);
\draw[->] (Target) -- (dec);

\end{tikzpicture}
	\caption{Deterministic identification for the Slow Fading Gaussian channel. The fading sequence remains constant during the transmission period.}
	\label{Fig.GaussianChannelSlow}
\end{figure}

\subsection{Coding with Slow Fading}
We move to the Gaussian channel with slow fading. In the compound channel model, we consider the worst-case channel and the error is maximized over the set of the values of the fading coefficients (see \cite[Sec. 23.3.1]{GK12}). As a result, we will show that the capacity is infinite as long as 
the set of fading values $\G$ does not include zero.

The definition of DI codes with CSI available at the decoder is given below.
\begin{definition}
\label{GdeterministicIDCode}
 An $(L(n,R),n)$ DI code  for a Gaussian channel $\sG_{\,\text{slow}}$ with CSI at the decoder, assuming $L(n,R)$ is an integer, is defined as a system $(\U,\mathscr{D})$ which consists of a codebook $\U=\{ \mathbf{u}_i \}_{i\in[\![L(n,R)]\!]}$, $\U\subset \mathbb{R}^n$,
 such that
 \begin{align} 
 \norm{\mathbf{u}_i}^2\leq nA \;,\,
 \end{align}
 for all $i\in[\![L(n,R)]\!]$ and a collection of decoding regions
\begin{align}
    \mathscr{D}=\{ \D_{i,g} \} \;,\,
\end{align}
for $i\in[\![L(n,R)]\!]$ and $g \in \G^n$ with
\begin{align}
    \bigcup_{i=1}^{L(n,R)} \D_{i,g} \subset \mathbb{R}^n \;.\,
\end{align}
 The error probabilities of the identification code $(\U,\mathscr{D})$ are given by
\begin{align}
     P_{e,1}(i) & = \sup_{g\in\G} \left[ 1- \int_{\D_{i,g}}\bigg( \prod_{t=1}^n f_{Z}(y_t-gu_{i,t}) \bigg) d\mathbf{y} \right] \;,\,
     \label{Eq.GTypeIErrorDefSlow}
     \\
     P_{e,2}(i,j) & =
     \sup_{g\in\G} \left[\int_{\D_{j,g}} \bigg(
     \prod_{t=1}^n f_{Z}(y_t-gu_{i,t})\bigg) d\mathbf{y} \right] \;,\,
     \label{Eq.GTypeIIErrorDefSlow}
\end{align}
with $f_{Z}(z)=\frac{1}{(2\pi\sigma_Z^2)^{1/2}} e^{-z^2/2\sigma_Z^2}$ (see Lemma~\ref{Fig.GaussianChannelSlow}).
An $(L(n,R),n,\lambda_1,\lambda_2)$ DI code is defined in a similar manner as for fast fading (see Section~\ref{GdeterministicIDCode})

A rate $R>0$ is called achievable if for every
$\lambda_1,\lambda_2>0$ and sufficiently large $n$, there exists an $(L(n,R),n,\lambda_1,\lambda_2)$ DI code. 
The operational DI capacity of the Gaussian channel is defined as the supremum of achievable rates, and will be denoted by $\mathbb{C}_{DI}(\mathscr{G}_{\,\text{slow}},L)$. 
\end{definition}

\begin{remark}
\label{Rem.Ginclude0}
If the fading coefficients can be zero or arbitrarily close to zero, i.e., $0\in \text{cl}(\G)$, then it immediately follows that the DI capacity is zero. To see this, observe that if $0\in \text{cl}(\G)$, then
\begin{align}
    P_{e,1}(i)+P_{e,2}(j,i) & \geq \left[ 1- \int_{\D_{i,g}} \left( \prod_{t=1}^n f_{Z}(y_t-gu_{i,t}) \right) d\mathbf{y} \right]_{g=0} + \left[ \int_{\D_{i,g}} \bigg( \prod_{t=1}^n f_{Z}(y_t-gu_{j,t})\bigg) d\mathbf{y} \right]_{g=0}
    \nonumber\\
    & = 1 \;.\,
\end{align}
\end{remark}
\subsection{Main Result - Slow Fading}
Our DI capacity theorem for the Gaussian channel with slow fading is stated below.
\begin{theorem}
\label{Th.GDICapacitySlow}
The DI capacity of the Gaussian channel $\sG_{\,\text{slow}}$ with slow fading in the super-exponential scale, i.e., for $L(n,R)=2^{n\log(n)R}$ is bounded by
\begin{align}
    \label{Eq.GDICapacitySlow1}
    \begin{array}{ll}
        \frac{1}{4}\leq\mathbb{C}_{DI}(\sG_{\,\text{slow}},L)\leq 1 & \text{ if } 0 \notin \text{cl}(\G)  \;,\,
        \\
        \mathbb{C}_{DI}(\sG_{\,\text{slow}},L)= 0 &\text{ if } 0 \in \text{cl}(\G) \;.\,
    \end{array}
\end{align}
Hence, the DI capacity is infinite in the exponential scale, if $0 \notin \text{cl}(\G)$,
\begin{align}
    \label{Eq.GDICapacitySlow2}
    & \mathbb{C}_{DI}(\sG_{\,\text{slow}},L) =
    \begin{cases}
        0 & \text{ if } 0 \in \text{cl}(\G) \;,\,
        \\
        \infty & \text{ if } 0 \notin \text{cl}(\G) \;,\,
    \end{cases}
\end{align}
and zero in the double exponential scale, i.e., for $L(n,R)=2^{2^{nR}}$, we have
\begin{align}
    \mathbb{C}_{DI}(\sG_{\,\text{slow}},L) = 0 \;.\, 
\end{align}
\end{theorem}
The derivation of the result above is similar to that of the proof for fast fading in Section~\ref{Sec.GaussianChannelFast}. For completeness, we give the proof of Theorem~\ref{Th.GDICapacitySlow} in Appendix~\ref{App.GDICapacitySlow}.
\section{Conclusion}
 \label{Sec.SummaryDiscussions}
 To summarize, we consider deterministic identification for Gaussian channels with fast fading and slow fading, where channel side information (CSI) is available at the decoder. We have shown that for Gaussian channels, the number of messages has a very different scale compared to standard results in both transmission and identification settings.
 In particular, the code size scales as $L(n,R)=2^{n\log(n)R}$. 
 We developed lower and upper bounds on the DI capacities $\mathbb{C}_{DI}(\sG_{\text{fast}},L)$ and $\mathbb{C}_{DI}(\sG_{\text{slow}},L)$ in this scale, for Gaussian channels with fast fading and slow fading, respectively. 
 Then, we deduced  that the DI capacity of a Gaussian Channel with fast fading is infinite in the exponential scale and zero in the double-exponential scale, regardless of the noise in the channel. That is,
 \begin{align*}
     \mathbb{C}_{DI}(\sG_{\text{fast}},L) = \mathbb{C}_{DI}(\sG_{\text{slow}},L) = \infty \;,\,
 \end{align*}
for $L(n,R)=2^{nR}$ and
\begin{align*}
    \mathbb{C}_{DI}(\sG_{\text{fast}},L) = \mathbb{C}_{DI}(\sG_{\text{slow}},L) = 0 \;,\,
\end{align*}
for $L(n,R)=2^{2^{nR}}$, provided that the fading coefficients are bounded away from zero.

We give two motivating examples for identification applications. In molecular communications (MC) \cite{NMWVS12,FYECG16}, information is transmitted via chemical signals or molecules. In various environments, e.g., inside the human body, conventional wireless communication with electromagnetic (EM) waves is not feasible or could be detrimental. The research on micro-scale MC for medical applications, such as intra-body networks, is still in its early stages and faces many challenges. MC is a promising contender for future applications such as 6G+ \cite{6G+}.
  
The results have the following geometric interpretation.
At first glance, it may seem reasonable that for the purpose of identification, one codeword could represent two messages.
While identification allows overlap between decoding regions \cite{AADT20,LDB20}, overlap at the encoder is not allowed for deterministic codes. We observe that when two messages are represented by codewords that are close to one another, then identification fails. If the probability of missed identification is upper bounded by $\epsilon_n$, then the  probability of false identification is lower bounded by $1-\delta_n$, where $\epsilon_n$ and $\delta_n$ tend to zero as $n\to\infty$. Hence, low  probability for the type I error comes at the expense of high probability for the type II error, and vice versa.
Thus, deterministic coding imposes the restriction that the codewords need to be distanced from each other. Based on fundamental properties in lattice and group theory \cite{CHSN13}, the optimal packing of non-overlapping spheres of radius $\sqrt{n\epsilon}$ contains an exponential number of spheres, and by decreasing the radius of the codeword spheres, the exponential rate can be made arbitrarily large.
However, in the  derivation of our lower bound, we show achievability of rates in the $2^{n\log(n)}$-scale by using spheres of radius  $\sqrt{n\epsilon_n}\sim n^{1/4}$, which results in $\sim 2^{\frac{1}{4}n\log(n)}$ codewords. 
The decoder in the achievability proof  performs a simple distance tests.

Alternatively, one may consider the $\epsilon$-capacity, for a fixed $0<\epsilon<1$. In the double exponential scale $L(n,R)=2^{2^{nR}}$, a rate $R$ is called $\epsilon$-achievable if there exists an
 $(L(n,R),n,\epsilon,\epsilon)$ code for sufficiently large $n$. The DI $\epsilon$-capacity $\mathbb{C}^{\epsilon}_{DI}(\sG,L)$ is then defined as the supremum of $\epsilon$-achievable rates.
As the DI and RI capacities in the double exponential scale have strong converses \cite{AD89,HV92,Bur94,MasterThesis},
\begin{align}
    {\mathbb{C}}_{RI}^{\epsilon}(\sG,L) &= {\mathbb{C}}_{RI}(\sG,L)=
    \frac{1}{2}\log\left(1+\frac{A}{\sigma_Z^2}\right) \;,\,
    \\
    {\mathbb{C}}_{DI}^{\epsilon}(\sG,L) &={\mathbb{C}}_{DI}(\sG,L)= 0 \;,\,
    \intertext{for  $0<\epsilon<\frac{1}{2}$. On the other hand, for $\epsilon\geq \frac{1}{2}$,\cite{BV00,AD89} we have}
 \label{Eq.InfiniteCapacity}
    {\mathbb{C}}_{DI}^{\epsilon}(\sG,L) &= {\mathbb{C}}_{RI}^{\epsilon}(\sG,L)= \infty \;.\,
\end{align}

\section*{Acknowledgments}
We gratefully thank Andreas Winter (Universitat Aut\`{o}noma de Barcelona) for a useful discussions concerning the scaling of the DI capacity. We thank Ning Cai (ShanghaiTech University) for a discussion on the DI capacity for DMCs. Finally, we thank Robert Schober (Friedrich Alexander University) for discussions and questions about the application of identification theory in molecular communications.

Mohammad J. Salariseddigh, Uzi Pereg, and Christian Deppe  were supported by 16KIS1005 (LNT, NEWCOM). Holger Boche was supported by 16KIS1003K (LTI, NEWCOM) and the BMBF within the national initiative for “Molecular Communications (MAMOKO)” under grant 16KIS0914.
\appendices
\section{Proof of Lemma~\ref{Lem.MiscScale}}
\label{App.MiscScale}
The proof is straightforward. Let $\mathbb{C}_{DI}(\sG_{\text{fast}}, L_0)=c_0$, where $c_0>0$ is a finite number. Then, for every $\lambda_1$, $\lambda_2$, and sufficiently large $n$, there exists an $(M_n,n,\lambda_1,\lambda_2)$ code where the number of message is  
\begin{align}
   \label{Eq.L_0}
   M_n=L_0(n,c_0-\epsilon) \;,\,
\end{align}
 where $\epsilon>0$ is arbitrarily small.


Assume to the contrary that in the $L^{-}$-scale, the DI capacity $\mathbb{C}_{DI}(\sG_{\text{fast}}, L^{-})=c^{-}$ is also finite. Then, the converse for this claim implies that the number of messages is bounded by
\begin{align}
   \label{Eq.L_-}
  M_n\leq L^{-}(n,c^{-}+\epsilon) \;.\,
\end{align}
Hence, by (\ref{Eq.L_0}) and (\ref{Eq.L_-}), 
\begin{align}
    L_0(n,c_0-\epsilon) \leq L^{-}(n,c^{-}+\epsilon) \;,\,
\end{align}
which contradicts the assumption that 
$L^{-}\prec L_0$ (see Definition~\ref{Def.DominatingScale}). Hence, $\mathbb{C}_{DI}(\sG_{\text{fast}}, L^{-})$ is infinite.

Assume to the contrary that in the $L^{+}$-scale, the DI capacity is positive, i.e. $\mathbb{C}_{DI}(\sG_{\text{fast}}, L^{+})=c^{+}>0$. Then, for every $\lambda_1$, $\lambda_2$, and sufficiently large $n$, there exists an $(M_n^{+},n,\lambda_1,\lambda_2)$ code where the number of messages is 
\begin{align}
   \label{Eq.L_+}
  M_n^{+}= L^{+}(n,c^{+}-\epsilon) \;.\,
\end{align}
Now, by the converse part for the $C_0$-scale,
\begin{align}
    M_n^{+}\leq L_0(n,c_0) \;.\,
    \label{Eq.M_+L0}
\end{align}
Hence, by (\ref{Eq.L_+}) and (\ref{Eq.M_+L0}), 
\begin{align}
    L^{+}(n,c^{+}-\epsilon) \leq L_0(n,c_0) \;,\,
\end{align}
which contradicts the assumption that 
$L_0\prec L^{+}$ (see Definition~\ref{Def.DominatingScale}). Hence, $\mathbb{C}_{DI}(\sG_{\text{fast}}, L^{+})$ is zero.
\section{Proof of Theorem~\ref{Th.GDICapacitySlow}}
\label{App.GDICapacitySlow}
\subsection{Lower Bound (Achievability Proof)}
Consider the Gaussian channel $\mathscr{G}_{\,\text{slow}}$ with slow fading. 
Based on Remark~\ref{Rem.Ginclude0}, when $0\in \text{cl}(\G)$, it immediately follows that the DI capacity is zero. Now, suppose that  $0\notin\text{cl}(\G)$. We show here that the DI capacity of the Gaussian channel with slow fading can be achieved using a dense packing arrangement and a simple distance-decoder.
  
A DI code for the Gaussian channel $\sG_{\,\text{slow}}$ with slow fading is constructed as follows.
Since the decoder can normalize the output symbols by $\frac{1}{\sqrt{n}}$, we have an equivalent input-output relation,
\begin{align}
    \bar{Y}_t = G\bar{x}_t + \bar{Z}_t \;,\,
\end{align}
where $G_t=G$ $\sim f_{G}$, and 
the noise sequence $\bar{\fZ}$ is i.i.d. $\sim \mathcal{N}\left(0,\frac{\sigma_Z^2}{n}\right)$, with an input power constraint 
\begin{align}
    \norm{\bar{\fx}} \leq \sqrt{A} \;,\,
\end{align}
with $\bar{\fx}=\frac{1}{\sqrt{n}}\fx$, $\bar{\fZ}=\frac{1}{\sqrt{n}}\fZ$, and $\bar{\fY}=\frac{1}{\sqrt{n}}\fY$.
\subsubsection*{Codebook construction}
As in our achievability proof for the fast fading setting (see Subsection~\ref{Subsec.CodebookConstructionGaussian}), we use a packing arrangement of non-overlapping hyper-spheres of radius $\sqrt{\epsilon_n}$ over a hyper-sphere of radius $(\sqrt{A}-\sqrt{\epsilon_n})$, with
\begin{align}
    \epsilon_n = \frac{A}{n^{\frac{1}{2}(1-b)}} \;,\,
\end{align}
where $b>0$ is an arbitrary small. As observed in Subsection~\ref{Subsec.AchievFast}, there exists an arrangement $$\bigcup_{i=1}^{2^{n\log(n)R}}\S_{\fu_i}(n,\sqrt{\epsilon_n}),$$ over $\S_{\f0}(n,\sqrt{A}-\sqrt{\epsilon_n})$ with a density of $\Delta_n\geq 2^{-n}$ \cite[Lem.~2.1]{C10}.
We assign a codeword to the center of each small sphere $\fu_i$. Since the small spheres have the same volume, the total number of spheres, i.e., the codebook size, satisfies
\begin{align}
    2^{n\log(n)R} &= \frac{\text{Vol}\left(\bigcup_{i=1}^{2^{n\log(n)R}}\S_{\fu_i}(n,\sqrt{\epsilon_n})\right)}{\text{Vol}(\S_{\fu_1}(n,\sqrt{\epsilon_n}))}
    \nonumber\\&
    \geq 2^{-n}\cdot\frac{\text{Vol}(\S_{\f0}(n,\sqrt{A}-\sqrt{\epsilon_n}))}{\text{Vol}(\S_{\fu_1}(n,\sqrt{\epsilon_n}))}
    \nonumber\\
    & =2^{-n}\cdot\left(\frac{\sqrt{A}-\sqrt{\epsilon_n}}{\sqrt{\epsilon_n}}\right)^n \;,\,
\end{align}
in a similar manner as in  Subsection~\ref{Subsec.CodebookConstructionGaussian}, 
hence,
\begin{align}
    \label{Eq.RateSlow}
    R\geq 
    \frac{1}{4}(1-b)-\frac{2}{\log(n)} \;,\,
\end{align}
which tends to $\frac{1}{4}$ when $n\to\infty$ and $b\to 0$.

\subsubsection*{Encoding}
Given a message $i\in [\![L(n,R)]\!]$, 
 transmit $\bar{\fx}=\bar{\fu}_i$.

\subsubsection*{Decoding}
Let
\begin{align}
    \delta_n & = \frac{\gamma^2 \epsilon_n}{3}
    \nonumber\\
    & = \frac{A\gamma^2}{3n^{\frac{1}{2}(1-b)}} \;,\,
    \label{Eq.deltaSlowDirect}
\end{align}
where $b>0$ is an arbitrary small. To identify whether a message $j\in [\![L(n,R)]\!]$ was sent, given the fading coefficient $g$, the decoder checks whether the channel output $\bar{\fy}$ belongs to the following decoding set,
\begin{align}
    \D_{j,g} = \left\{ \bar{\fy}\in\mathbb{R}^n \,:\; \sum_{t=1}^n (\bar{y}_t-g    \bar{u}_{j,t})^2
    \leq \sqrt{\sigma_Z^2+\delta_n} \right\} \;.\,
\end{align}
\subsubsection*{Error Analysis}
Consider the type I error, i.e., when the transmitter sends $\bar{\fu}_i$, yet $\bar{\fY}\notin\D_{i,G}$. For every $i\in[\![L(n,R)]\!]$, the type I error probability is given by
\begin{align}
        P_{e,1}(i) & = \sup_{g\in\G} \left[ P_{e,1} \left( i \,| g \; \right) \right] \;,\,
        \label{Eq.Pe2GSlow0}
\end{align}
where we have defined
\begin{align}
    P_{e,1} \left( i \,| g \; \right) & \equiv \Pr\left(\sum_{t=1}^n (\bar{Y}_t-G\bar{u}_{i,t})^2 > \sigma_Z^2+\delta_n \,\big|\,\bar{\fx}=(\bar{u}_{i,t})_{t=1}^n \,,\, G = g \right)
    \nonumber\\
    & = \Pr\left(\sum_{t=1}^n {\bar{Z}_t}^2 > \sigma_Z^2 + \delta_n \right)
    \;,\,
\end{align}
for $g\in\G$, as the fading coefficient $G$ and the noise vector $\bar{\fZ}$ are statistically independent.

Now we can bound the type I error probability by
\begin{align}
    P_{e,1} \left( i \,| g \; \right) & = \Pr\left( \sum_{t=1}^n {\bar{Z}_t}^2 - \sigma_Z^2 > \delta_n \right)
    \nonumber\\
    & \leq \frac{3\sigma_Z^4}{n\delta_n^2}
    \nonumber\\
    & = \frac{27\sigma_Z^4}{A^2 \gamma^4 n^b}
    \nonumber\\
    & \leq \lambda_1 \;,\,
\end{align}
for sufficiently large $n$ and arbitrarily small $\lambda_1>0$, where the first inequality holds by Chebyshev's inequality and since the fourth moment of a Gaussian variable $V\sim \N(0,\sigma_V^2)$ is $\mathbb{E}\{V^4\}=3\sigma_V^4$. Thus we have $P_{e,1} \left( i \, | g \,\right) \leq \lambda_1$ for all $g\in\G$.
Hence, the type I error probability satisfies $P_{e,1}\left( i \right) \leq \lambda_1$ (see (\ref{Eq.Pe2GSlow0})).

Next we address the type II error, i.e., when $\bar{\fY}\in\D_{j,G}$ while the transmitter sent $\bar{\fu}_i$.
Then, for every $i,j\in[\![L(n,R)]\!]$, where $i\neq j$, the type II error probability is given by
\begin{align}
    P_{e,2}(i,j) & = \sup_{g\in\G} \left[ P_{e,2} \left( i,j \,| g \; \right) \right] \;,\,
    \label{Eq.Pe2GSlow1}
\end{align}
where we have defined
\begin{align}
    P_{e,2}\left( i,j \,| g \; \right) & \equiv \Pr\left(\sum_{t=1}^n (\bar{Y}_t - G\bar{u}_{j,t})^2
    \leq \sigma_Z^2+\delta_n \,\big|\,\bar{\fx}=(\bar{u}_{i,t})_{t=1}^n \,,\, G = g \right)
    \nonumber\\
    & = \Pr\left(\sum_{t=1}^n (g(\bar{u}_{i,t}-\bar{u}_{j,t})+\bar{Z}_t)^2
    \leq \sigma_Z^2+\delta_n \right)
    \label{Eq.Pe2GSlow}
\end{align}
for $g\in\G$, as the fading coefficient $G$ and the noise vector $\bar{\fZ}$ are statistically independent.
Now we bound the probability within the square brackets. 

We divide into two cases. First, consider $g \in \G$ such that $\norm{g \left( \bar{\fu}_i - \bar{\fu}_j \right)} > 2 \sqrt{\sigma_Z^2 + \delta_n}$.
Therefore, by the reverse triangle inequality, $\norm{\fa - \fb} \geq \left| \norm{\fa} - \norm{\fb} \right|$, we have
\begin{align}
    \sqrt{ \sum_{t=1}^n \left( \left( g\left( \bar{u}_{i,t} - \bar{u}_{j,t} \right) \right) + \bar{Z}_t \right)^2} & \geq \norm{g \left( \bar{\fu}_i - \bar{\fu}_j \right)} - \norm{\bar{\fZ}}
    \nonumber\\
    & \geq 2 \sqrt{\sigma_Z^2 + \delta_n} - \norm{\bar{\fZ}} \,.
\end{align}
Hence, for every $g$ such that $\norm{g \left( \bar{\fu_i} - \bar{\fu_j} \right)} > 2 \sqrt{\sigma_Z^2 + \delta_n}$, we can bound the type II error probability by 
\begin{align}
    \label{Ineq.TypeII_Slow_Conditional1}
    P_{e,2} \left( i,j \,\big| g \,\right)
    & \leq \Pr\left( \norm{\bar{\fZ}} \geq \sqrt{\sigma_Z^2 + \delta_n} \right)
    \nonumber\\
    & = \Pr\left(\sum_{t=1}^n {\bar{Z}_t}^2 > \sigma_Z^2 + \delta_n \right)
    \nonumber\\
    & \leq \frac{3\sigma_Z^4}{n\delta_n^2} 
    \nonumber\\
    & = \frac{27\sigma_Z^4}{n^b A^2 \gamma^4}
    \nonumber\\
    & \leq \lambda_2 \;,\,
\end{align}
for sufficiently large  $n$ and arbitrarily small $\lambda_2>0$, where the second inequality follows by Chebyshev inequality.

Now, we turn to the second case, i.e., when
\begin{align}
    \label{Ineq.Bound_G}
    \norm{g \left( \bar{\fu}_i - \bar{\fu}_j \right)} \leq 2 \sqrt{\sigma_Z^2 + \delta_n} \;.\,
\end{align}
Observe that for every given $g\in\G$, 
\begin{align}
    \sum_{t=1}^n (g(\bar{u}_{i,t}-\bar{u}_{j,t}) + \bar{Z}_t)^2=
    \sum_{t=1}^n g^2(\bar{u}_{i,t}-\bar{u}_{j,t})^2+\sum_{t=1}^n \bar{Z}_t^2+2\sum_{t=1}^n g(\bar{u}_{i,t}-\bar{u}_{j,t})Z_t \;.\,
     \label{Eq.Pe2normSlow}
\end{align}
Then define the event
\begin{align}
    \E_0(g) = \left\{ \Big|\sum_{t=1}^n g(\bar{u}_{i,t}-\bar{u}_{j,t})\bar{Z}_t\Big|>\frac{\delta_n}{2} \right\} \;.\,
    \label{Eq.E0FadingSlow}
\end{align}
By Chebyshev's inequality, the probability of this event vanishes, 
\begin{align}
    \Pr\left( \E_0(g) \right) & \leq \frac{g^2\sum_{t=1}^n (\bar{u}_{i,t}-\bar{u}_{j,t})^2\mathbb{E}\{\bar{Z}_t^2\}}{\left(\frac{\delta_n}{2} \right)^2}
    \nonumber\\
    & = \frac{4 \sigma^2_Z \norm{g \left( \bar{\fu}_i-\bar{\fu}_j \right)}^2}{n\delta_n^2}
    \nonumber\\
    & \leq \frac{16\sigma^2_Z \left( \sigma_Z^2 + \delta_n \right) }{n\delta_n^2}
    \nonumber\\
    & \leq \tau_0 \;,\,
    \label{Eq.PeE0G1Slow}
\end{align}
for sufficiently large $n$ and arbitrarily small $\tau_0 > 0$, where the first inequality holds since the sequence $\{\bar{Z}_t\}$ is i.i.d. $\sim\mathcal{N}\left(0,\frac{\sigma_Z^2}{n}\right)$, and the second inequality follows from (\ref{Ineq.Bound_G}).
Furthermore, observe that given the complementary event $\E_0^c(g)$, we have
\begin{align}
    2\sum_{t=1}^n g\left( \bar{u}_{i,t} - \bar{u}_{j,t} \right) \bar{Z}_t \geq - \delta_n \;,\,
\end{align}
Therefore, the event $\E_0^c$, the type II error event in (\ref{Eq.Pe2GSlow}), and the identity in (\ref{Eq.Pe2normSlow}) together imply that the following event occurs,
\begin{align}
    \E_1 (g) & = \left\{ \sum_{t=1}^n g^2(\bar{u}_{i,t}-\bar{u}_{j,t})^2 + \sum_{t=1}^n \bar{Z}_t^2\leq \sigma_Z^2+2\delta_n \right\} \;.\,
    \label{Eq.Pe2normConsequenceSlow}
\end{align}
Now lets define
\begin{align}
    \H_{i,j}^n = \left\{ \fG \in \G^n \;:\, \sum_{t=1}^n (g(\bar{u}_{i,t}-\bar{u}_{j,t})+\bar{Z}_t)^2 \leq \sigma_Z^2+\delta_n \right\} \;.\,
\end{align}
Therefore, applying the law of total probability to (\ref{Eq.Pe2Bound1Slow}), we have
\begin{align}
    P_{e,2} \left( i,j \,\big| g \,\right) & =
    \Pr\left( \H_{i,j}^n \cap \E_0(g) \right) + \Pr\left( \H_{i,j}^n \cap {\E_0^c(g)} \right)
    \nonumber\\
    & \leq \Pr(\E_0(g)) + \Pr\left( \E_1(g) \right)
    \nonumber\\
    & \leq \tau_0 + \Pr\left( \E_1(g) \right) \;,\,
    \label{Eq.Pe2Bound1Slow}
\end{align}
where the last inequality holds by (\ref{Eq.PeE0G1Slow}).

Now we focus on the second term in (\ref{Eq.Pe2Bound1Slow}), i.e., $\Pr(\E_1(g))$. To this end, observe that based on the codebook construction, each codeword is surrounded by a sphere of radius $\sqrt{\epsilon_n}$, which implies that
\begin{align}
    \norm{\bar{\fu}_i-\bar{\fu}_j}\geq \sqrt{\epsilon_n} \;.\,
\end{align}
Thus, we have
\begin{align}
    g^2 \norm{\bar{\fu}_i - \bar{\fu}_j}^2\geq  \gamma^2 \epsilon_n \;,\,
\end{align}
where
$\gamma$ is the minimal value in ${\G}$. Hence, according to (\ref{Eq.Pe2Bound1Slow}),
\begin{align}
    P_{e,2} \left( i,j \,\big| g \,\right)
    & \leq \Pr\left( \norm{\bar{\fZ}}^2\leq \sigma_Z^2+2\delta_n-\gamma^2 \epsilon_n  \right) + \tau_0
    \nonumber\\
    & =
    \Pr\left( \norm{\bar{\fZ}}^2-\sigma_Z^2\leq -\delta_n \right)+\tau_0 \;,\,
\end{align}
(see (\ref{Eq.deltaSlowDirect})). Therefore, by Chebyshev's inequality, 
\begin{align}
    \label{Ineq.TypeII_Slow_Conditional2}
    P_{e,2} \left( i,j \,\big| g \,\right) & \leq \Pr\left(\sum_{t=1}^n \bar{Z}_t^2-\sigma_Z^2 \leq -\delta_n
    \right) + \tau_0
    \nonumber\\
    & \leq \frac{\sum_{t=1}^n\text{var}(\bar{Z}_t^2)}{\delta_n^2} + \tau_0
    \nonumber\\
    & \leq \frac{\sum_{t=1}^n \mathbb{E}\{\bar{Z}_t^4\}}{\delta_n^2} + \tau_0
    \nonumber\\
    & = \frac{ 3n \left(\frac{\sigma_Z^2}{n}\right)^2}{\delta_n^2} + \tau_0
    \nonumber\\
    & = \frac{3\sigma_Z^4}{n\delta_n^2} + \tau_0
    \nonumber\\
    & = \frac{27\sigma_Z^4}{n^bA^2\gamma^4} + \tau_0
    \nonumber\\
    & \leq \lambda_2 \;,\,
\end{align}
for sufficiently large $n$.
Based on (\ref{Ineq.TypeII_Slow_Conditional1}) and (\ref{Ineq.TypeII_Slow_Conditional2}),
we have $P_{e,2} \left( i,j \, | g \,\right) \leq \lambda_2$ for all $g\in\G$.
Hence, the type II error probability satisfies $P_{e,2}\left( i, j \right) \leq \lambda_2$ (see (\ref{Eq.Pe2GSlow1})).

We have thus shown that for every $\lambda_1,\lambda_2>0$ and sufficiently large $n$, there exists a $(2^{n\log (n) R}, n, \lambda_1, \lambda_2)$ code.
As we take the limits of $n\rightarrow\infty$, and then $b \rightarrow 0$, the lower bound on the achievable rate tends to $\frac{1}{4}$, by (\ref{Eq.RateSlow}). This completes the achievability proof for Theorem~\ref{Th.GDICapacitySlow}.
\subsection{Upper Bound (Converse Proof)}
\label{Subsec.DoubleExponentialSlow}
Suppose that $R$ is an achievable rate in the $L$-scale for the Gaussian channel with slow fading. 
Consider a sequence of $(L(n,R),n,\lambda_1^{(n)},\lambda_2^{(n)})$ codes $(\U^{(n)},\D^{(n)})$, such that
$\lambda_1^{(n)}$ and $\lambda_2^{(n)}$ tend to zero as $n\rightarrow\infty$. We begin with the following lemma.
\begin{lemma}
\label{Lem.DConverseSlow}
Consider a sequence of codes  as described above.
Let $b>0$ be an arbitrarily small constant that does not depend on $n$.
%
 Then there exists $n_0(b)$, such that for all $n>n_0(b)$, every pair of codewords in the codebook $\U^{(n)}$ are distanced by at least $\sqrt{n\epsilon_n}$, i.e.,
    \begin{subequations}
    \begin{align}
     \norm{\fu_{i_1} - \fu_{i_2}}\geq \sqrt{n\epsilon_n} \;,\,
    \end{align}
    where
    \begin{align}
        \epsilon_n = \frac{A}{n^{2(1+b)}} \;,\,
    \end{align}
    \end{subequations}
for all $i_1,i_2\in [\![L(n,R)]\!]$, such that $i_1\neq i_2$.
\end{lemma}
\begin{proof}
Fix $\lambda_1$ and $\lambda_2$. Let $\kappa,\theta,\zeta>0$ be arbitrarily small. Assume to the contrary that there exist two messages $i_1$ and $i_2$, where $i_1\neq i_2$, such that
    \begin{align}
        \norm{\fu_{i_1} - \fu_{i_2}} < \sqrt{n\epsilon_n} = \alpha_n \;,\,
        \label{Eq.distanceCodeSlow}
    \end{align}
    where
    \begin{align}
        \label{Eq.Alpha_nSlow1}
        \alpha_n\equiv \frac{\sqrt{A}}{n^{\frac{1}{2}(1+2b)}} \;.\,
    \end{align}
    Now let us define two subsets as follows
    \begin{align}
        \mathfrak{B}_{i_1,i_2} = \left\{ \fy\in\D_{i_1,g} \,:\; \sum \limits_{t=1}^n \left( y_t - gu_{{i_2},t} \right)^2 \leq n\left( \sigma_Z^2 + \zeta \right) \right\}
        \label{Eq.Event_B_Slow}
    \end{align}
    \begin{align}
        \mathfrak{C}_{i_1,i_2} & = \left\{ \fy \in \Y^n \,:\; \sum \limits_{t=1}^n \left( y_t - gu_{{i_2},t} \right)^2 \leq n\left( \sigma_Z^2 + \zeta \right) \right\} \;.\,
        \label{Eq.Event_C_Slow}
    \end{align}
    Observe that for every $g\in\G$,
    \begin{align}
        \int_{\D_{i_1,g}} \left(\prod_{t=1}^n f_{Z} \left( y_t - gu_{i_1,t} \right) \right) d\fy & = \int_{\mathfrak{B}_{i_1,i_2}}
        \left(\prod_{t=1}^n f_{Z} \left( y_t - gu_{i_1,t} \right) \right) \, d\fy + \int_{\D_{i,1,g} \setminus \mathfrak{B}_{i_1,i_2}} \left(\prod_{t=1}^n f_{Z} \left( y_t - gu_{i_1,t} \right) \right) \, d\fy
        \nonumber\\
        & \leq \int_{\mathfrak{B}_{i_1,i_2}} \left(\prod_{t=1}^n f_{Z} \left( y_t - gu_{i_1,t} \right) \right) \, d\fy + \int_{\mathfrak{C}_{i_1,i_2}^c} \left( \prod_{t=1}^n f_{Z} \left (y_t - gu_{i_1,t} \right) \right) d\fy \;.\,
        \label{Eq.Pe1boundConv0Slow}
    \end{align}
    where the last inequality holds since
    \begin{align}
        \mathfrak{C}_{i_1,i_2}^c \supset \D_{i_1,g} \setminus \mathfrak{B}_{i_1,i_2} \;,\,
    \end{align}
     with $\setminus$ being the set minus operation. Consider the second integral, for which the domain is $\mathfrak{C}_{i_1,i_2}^c$ and where we denote
        \begin{align}
            \fg\equiv (g,g,\ldots, g) \;.\,
        \end{align}
        Then, by the triangle inequality,
        \begin{align}
            \norm{\fy-\fg\circ\fu_{i,1}}&\geq
            \norm{\fy-\fg\circ\fu_{i,2}}-\norm{\fg\circ(\fu_{i,1}-\fu_{i,2})}
            \nonumber\\
            &=
            \norm{\fy-\fg\circ\fu_{i,2}}-g\norm{\fu_{i,1}-\fu_{i,2}}
            \nonumber\\
            &> \sqrt{n(\sigma_Z^2+\zeta)}-g\norm{\fu_{i,1}-\fu_{i,2}}
            \nonumber\\ &\geq\sqrt{n(\sigma_Z^2+\zeta)}-g\alpha_n \;.\,
        \end{align}
    For sufficiently large $n$, this implies the following subset
        \begin{align}
           \mathfrak{F}_{i_1,i_2}^c = \left\{y^n \in \Y^n \; : \, \norm{\fy - \fg \circ\fu_{i,1}} > \sqrt{n\left( \sigma_Z^2 + \eta \right)} \right\} \;,\,
           \label{Eq.Regiong0_Slow}
        \end{align}
    for  $\eta<\frac{\zeta}{2}$. That is,
    \begin{align}
        \left\{\fy \in \Y^n \; : \, \norm{\fy-\fg\circ\fu_{i,2}} \geq
        \sqrt{n\left( \sigma_Z^2 + \zeta \right)} \right\} \quad \overset{\text{implies}}{\longrightarrow} \quad \left\{\fy \in \Y^n \; : \, \norm{\fy-\fg\circ\fu_{i,1}} \geq
        \sqrt{n \left( \sigma_Z^2 + \eta \right)} \right\} \;.\,
    \end{align}
    Thus we deduce that
        \begin{align}
            \mathfrak{F}_{i_1,i_2}^c \supset \mathfrak{C}_{i_1,i_2}^c \;,\,
        \end{align}
    Hence, the second integral in the right hand side of (\ref{Eq.Pe1boundConv0Slow}) is bounded by
    \begin{align}
        \int_{\mathfrak{F}_{i_1,i_2}^c} f_{\fZ}(\fy-\fg\circ\fu_{i_1}) d\fy & = \Pr\left( \norm{\fy-\fg\circ\fu_{i,1}} \geq
        \sqrt{n \left( \sigma_Z^2 + \eta \right)} \right) 
        \nonumber\\
        & = \Pr(\norm{\fZ}^2-n\sigma_Z^2>n\eta)
        \nonumber\\
        & \leq \frac{3\sigma_Z^4}{n\eta^2}
        \nonumber\\
        & \leq \kappa \;,\,
    \end{align}
    for sufficiently large $n$, where the third line is due to Chebyshev's inequality, followed by the substitution of $\fz\equiv \fy-\fg\circ\fu_{i_1} $.
   Thus, by (\ref{Eq.Pe1boundConv0Slow}),
    \begin{align}
        \label{Eq.ComplTypeISlow}
        \int_{\D_{i_1,g}} \left( \prod_{t=1}^n f_{Z} \left( y_t - g u_{i_1,t} \right) \right) \, d\fy \leq  \int_{\mathfrak{B}_{i_1,i_2}} f_{\fZ} \left( \fy - \fg \circ \fu_{i_1} \right) \, d\fy + \kappa \;.\,
    \end{align}
    Now, we can focus on the first integral with domain of $\mathfrak{B}_{i_1,i_2}$, i.e., when
    \begin{align}
        \norm{\fy-\fg\circ\fu_{i,2}}\leq\sqrt{n(\sigma_Z^2+\zeta)} \;.\,
        \label{Eq.ui2DistSlow}
    \end{align}
    Observe that
    \begin{align}
        f_{\fZ}(\fy-\fg\circ\fu_{i_1}) - f_{\fZ}(\fy-\fg\circ\fu_{i_2})=  f_{\fZ}(\fy-\fg\circ\fu_{i_1})\left[1-e^{-\frac{1}{2\sigma_Z^2}\left(\norm{\fy-\fg\circ\fu_{i_2}}^2-\norm{\fy-\fg\circ\fu_{i_1}}^2\right)}\right] \;.\,
    \end{align}
    By the triangle inequality,
    \begin{align}
        \label{Ineq.triangleSlow}
        \norm{\fy-\fg\circ\fu_{i_1}}\leq  \norm{\fy-\fg\circ\fu_{i_2}} +  g\norm{\fu_{i_1} - \fu_{i_2}} \;.\,
    \end{align}
    Taking the square of both sides, we have
    \begin{align}
        \norm{\fy-\fg\circ\fu_{i_1}}^2 &\leq \norm{\fy-\fg\circ\fu_{i_2}}^2 +  g^2\norm{\fu_{i_2} - \fu_{i_1}}^2 +
        2\norm{\fy-\fg\circ\fu_{i_2}}\cdot g\norm{\fu_{i_2} - \fu_{i_1}}
        \nonumber\\
        & \leq
        \norm{\fy-\fg\circ\fu_{i_2}}^2+ g^2\alpha_n^2+ 2g\alpha_n\sqrt{n(\sigma_Z^2+\zeta)}
        \nonumber\\
        &=
       \norm{\fy-\fg\circ\fu_{i_2}}^2+ g^2\alpha_n^2+
       2g\frac{\sqrt{A(\sigma_Z^2+\zeta)}}{n^b} \;,\,
    \end{align}
    where the second line follows from (\ref{Eq.distanceCodeSlow}) and (\ref{Eq.ui2DistSlow}), and the line is due to (\ref{Eq.Alpha_nSlow1}).
    Thus, for sufficiently large $n$,
    \begin{align}
        \norm{\fy-\fg\circ\fu_{i_1}}^2-\norm{\fy-\fg\circ\fu_{i_2}}^2\leq \theta \;.\,
    \end{align}
    Hence,
    \begin{align}
        \label{Ineq.GaussianContinuitySlow}
        f_{\fZ}(\fy-\fg\circ\fu_{i_1}) - f_{\fZ}(\fy-\fg\circ\fu_{i_2})&\leq  f_{\fZ}(\fy-\fg\circ\fu_{i_1})\left(1-e^{-\frac{\theta}{2\sigma_Z^2}}\right)
        \nonumber\\
        & \leq \kappa f_{\fZ}\left(\fy-\fg\circ\fu_{i_1}\right) \;,\,
    \end{align}
    for sufficiently small $\theta > 0$, such that $1-e^{-\frac{\theta}{2\sigma_Z^2}} \leq \kappa$. Now by (\ref{Eq.ComplTypeISlow}), we get,
    \begin{align}
       \lambda_1 + \lambda_2 & \geq 
       P_{e,1}(i_1) + P_{e,2}(i_2,i_1)
       \nonumber\\
       & \stackrel{(a)}{\geq} \sup_{g\in\G} \left[ P_{e,1}(i_1 | g) \right] + \sup_{g\in\G} \left[ P_{e,2}(i_2,i_1 | g) \right]
       \nonumber\\
       & \stackrel{(b)}{\geq} \sup_{g\in\G} \left[ P_{e,1}(i_1 | g) + P_{e,2}(i_2,i_1 | g) \right]
      \nonumber\\
      & \geq \sup_{g\in\G} \left[ 1 - \int_{\D_{i_1,g}} \left( \prod_{t=1}^n f_{Z} \left( y_t - g u_{i_1,t} \right) \right) \, d \mathbf{y} + \int_{\D_{i_1,g}} \left( \prod_{t=1}^n f_{Z} \left( y_t - g u_{i_2,t} \right) \right) \, d\fy \right]
       \nonumber\\
      & \stackrel{(c)}{\geq} \sup_{g\in\G} \left[ 1 - \kappa -
      \int_{\mathfrak{B}_{i_1,i_2}} f_{\fZ} \left( \fy - \fg \circ \fu_{i_1} \right) \, d\fy +
      \int_{\D_{i_1,g}} f_{\fZ}(\fy-\fg\circ\fu_{i_2}) \, d\fy \right]
      \nonumber\\
      & \geq \sup_{g\in\G} \left[ 1 - \kappa - \int_{\mathfrak{B}_{i_1,i_2}} f_{\fZ} \left( \fy - \fg \circ \fu_{i_1} \right) \, d\fy + \int_{\mathfrak{B}_{i_1,i_2}} f_{\fZ} \left( \fy - \fg \circ \fu_{i_2} \right) \, d\fy \right]
      \nonumber\\
      & = \sup_{g\in\G} \left[1 - \kappa - \int_{\mathfrak{B}_{i_1,i_2}} \left( f_{\fZ} \left( \fy - \fg \circ \fu_{i_1} \right) - f_{\fZ} \left( \fy - \fg \circ \fu_{i_2} \right) \right) \, d\fy \right] \;,\,
    \end{align}
    where $(a)$ follows by definitions given in (\ref{Eq.Pe2GSlow0}) and (\ref{Eq.Pe2GSlow1}), $(b)$ holds since supremum is sub-additive and $(c)$ is due to (\ref{Eq.ComplTypeISlow}). Hence, by (\ref{Ineq.GaussianContinuitySlow}),
    \begin{align}
        \lambda_1 + \lambda_2 
        & \geq 1 - \kappa - \kappa \inf_{g\in\G} \left[ \int_{\mathfrak{B}_{i_1,i_2}} f_{\fZ} \left( \fy - \fg \circ \fu_{i_1} \right) \, d\fy \right]
        \nonumber\\
        & \geq 1 - 2\kappa \;,\,
    \end{align}
    which leads to a contradiction for sufficiently small $\kappa$ such that $2\kappa
    <1-\lambda_1-\lambda_2$. This completes the proof of Lemma~\ref{Lem.DConverseSlow}.
    \end{proof}
By Lemma~\ref{Lem.DConverseSlow} we can define an arrangement of non-overlapping spheres $\S_{\fu_i}(n,\sqrt{n\epsilon_n})$ of radius $\sqrt{n\epsilon_n}$ centered at the codewords $\fu_i$.
Since the codewords all belong to a sphere $\S_{\f0}(n,\sqrt{nA})$ of radius $\sqrt{nA}$ centered at the origin, it follows that the number of packed spheres, i.e., the number of codewords $2^{n\log(n)R}$, is bounded by
\begin{align}
    2^{n\log(n)R} & \leq \frac{\text{Vol}(\S_{\f0}(n,\sqrt{nA}+\sqrt{n\epsilon_n}))}{\text{Vol}(\S_{\fu_i}(n,\sqrt{n\epsilon_n}))}
    \nonumber\\
    & = \left(\frac{\sqrt{A}+\sqrt{\epsilon_n}}{\sqrt{\epsilon_n}}\right)^n \;.\,
\end{align}
Thus,
\begin{align}
     R &\leq \frac{1}{\log n} \log\left(\frac{\sqrt{A}+\sqrt{\epsilon_n}}{\sqrt{\epsilon_n}}\right)
    \nonumber\\
    & = 1 + b + \frac{\log \left( 1 + \frac{1}{n^{1 + b}} \right) }{\log n} \;,\,
\end{align}
by the same arguments as in the proof for the Gaussian channel with fast fading (see (\ref{Ineq.Rate_Fast})), which tends to $1+b$ as $n \to \infty$. Now, since $b>0$ is arbitrarily small, an achievable rate must satisfy $R \leq 1$. This completes the proof of Theorem~\ref{Th.GDICapacitySlow}.
\qed
\bibliographystyle{IEEEtran}
\bibliography{IEEEabrv,confs-jrnls,Ref}
\end{document}